\documentclass[a4paper,UKenglish,cleveref,autoref,nameinlink,thm-restate]{lipics-v2021}

\usepackage{apptools}
\usepackage{graphicx}
\usepackage{comment}
\usepackage{paralist}
\usepackage{todonotes}
\usepackage{orcidlink}
\usepackage{thmtools} 
\usepackage{thm-restate}

\usepackage{amsmath}

\usepackage{hyperref}
\usepackage[capitalise]{cleveref}
\usepackage{complexity}
\usepackage[inline]{enumitem}

\hideLIPIcs  %uncomment to remove references to LIPIcs series (logo, DOI, ...), e.g. when preparing a pre-final version to be uploaded to arXiv or another public repository

\graphicspath{{./figures/}}

\newcommand{\restateref}[1]{\IfAppendix{\hyperref[#1]{$\star$}}{\hyperref[#1*]{$\star$}}}

\newtheorem{lemma}[theorem]{Lemma}

\Crefname{property}{Prop.}{Props.}
\Crefname{theorem}{Thm.}{Thms.}
\Crefname{section}{Sect.}{Sects.}
\Crefname{figure}{Fig.}{Figs.}

\nolinenumbers %uncomment to disable line numbering

\usepackage{overarrows}
\newcommand*{\leftarrowhead}{\usefont{U}{lasy}{m}{n}\symbol{40}}
\newcommand*{\righttarrowhead}{\usefont{U}{lasy}{m}{n}\symbol{41}}

\NewOverArrowCommand{firstleg}{%
start=\vcenter{\hbox{$\smallermathstyle\circ$}},
end=\cdot,
trim start=0.7, trim end=0.7,
min length=20,
shift leftright=-2,
}

\NewOverArrowCommand{secondleg}{%
start=\cdot,
end=\cdot,
trim start=0.1, trim end=0.1,
min length=20,
shift leftright=-2,
}

\NewOverArrowCommand{thirdleg}{%
start=\cdot,
end=\vcenter{\hbox{$\smallermathstyle\circ$}},
trim start=0.7, trim end=0.7,
min length=20,
shift leftright=-2,
}
\NewOverArrowCommand{overrightleftarrow}{% not used
start=\text{\righttarrowhead},
end=\text{\leftarrowhead},
trim start=0.7, trim end=0.7,
min length=20,
shift leftright=-2,
}

\newcommand*{\children}[0]{\ensuremath{\text{children}}}

\newenvironment{sketch}{\medskip\noindent {\emph{Proof sketch.}}}{\hfill\qed}
\renewenvironment{proof}{\medskip\noindent {\emph{Proof.}}}{\hfill\qed}

\title{Tangling and Untangling Trees on Point-sets}
\titlerunning{Tangling and Untangling Trees on Point-sets}
\author{Giuseppe Di Battista}{Roma Tre University, Rome, Italy}{giuseppe.dibattista@uniroma3.it}{https://orcid.org/0000-0003-4224-1550}{}

\author{Giuseppe Liotta}{University of Perugia, Perugia, Italy}{giuseppe.liotta@unipg.it}{https://orcid.org/0000-0002-2886-9694}{}

\author{Maurizio Patrignani}{Roma Tre University, Rome, Italy}{maurizio.patrignani@uniroma3.it}{https://orcid.org/0000-0001-9806-7411}{}

\author{Antonios Symvonis}{School of Applied Mathematical \& Physical Sciences, NTUA, Greece}{symvonis@math.ntua.gr}{https://orcid.org/0000-0002-0280-741X}{}

\author{Ioannis G. Tollis}{Department of Computer Science, University of Crete, Heraklion, Greece}{tollis@csd.uoc.gr}{https://orcid.org/0000-0002-5507-7692}{}

\authorrunning{Di Battista, Liotta, Patrignani, Symvonis, Tollis}

\Copyright{Giuseppe Di Battista, Giuseppe Liotta, Maurizio Patrignani, Antonios Symvonis, Ioannis G. Tollis}

\keywords{Tree drawings, Prescribed edge crossings, Curve complexity, Point-set embeddings, RAC drawings, Topological linear embeddings}

\ccsdesc[500]{Mathematics of computing~Graph algorithms}
\ccsdesc[500]{Mathematics of computing~Graph theory}

\begin{document}
\maketitle

\begin{abstract}
We study a question that lies at the intersection of classical research subjects in Topological Graph Theory and Graph Drawing: Computing a drawing of a graph with a prescribed number of crossings on a given set $S$ of points, while ensuring that its curve complexity (i.e., maximum number of bends per edge) is bounded by a constant. 
We focus on trees: Let $T$ be a tree, $\vartheta(T)$ be its thrackle number, and $\chi$ be any integer in the interval $[0,\vartheta(T)]$. In the \emph{tangling} phase we compute a topological linear embedding of $T$ with $\vartheta(T)$ edge crossings and a constant number of spine traversals. In the \emph{untangling} phase we remove edge crossings  without increasing the  spine traversals until we reach $\chi$ crossings. The computed linear embedding is used to construct a drawing of $T$ on $S$ with $\chi$ crossings and constant curve complexity. Our approach gives rise to an $O(n^2)$-time algorithm for general trees and an $O(n \log n)$-time algorithm for paths. We also adapt the approach to compute RAC drawings, i.e. drawings where the angles formed at edge crossings are~$\frac{\pi}{2}$. 

\end{abstract}

%---------------------------------------------------------------
\section{Introduction}\label{sec:introduction}
%---------------------------------------------------------------

In this paper we study the following problem:
Given an $n$-vertex tree $T$, a set $S$ of $n$ distinct points in the plane, and a non-negative integer $\chi$, find a simple drawing $\Gamma$ of $T$ with the following properties: 
\begin{inparaenum}[(a)]
\item\label{prop:pointeset}$\Gamma$ is a \emph{point-set embedding} of $T$ on $S$, i.e. its vertices are the points of $S$,
\item\label{prop:curve-complexity}the edges of $T$ are represented by polylines with constant \emph{curve complexity}, measured as the maximum number of bends per edge, and
\item\label{prop:crossings}the edges cross exactly $\chi$ times in~$\Gamma$.
\end{inparaenum}
In the following we refer to them as Property~\ref{prop:pointeset}, Property~\ref{prop:curve-complexity}, and Property~\ref{prop:crossings}.

The problem locates itself at the intersection of three well studied research subjects in Graph Drawing and Graph Theory, namely
the study of 
point-set embeddings,
the study of 
drawings with limited curve complexity,
and 
the study of 
drawings with prescribed numbers of edge crossings.
Before describing our contribution, we briefly recall each such topic.

\noindent Property \ref{prop:pointeset} -- For a graph $G$ with $n$ vertices and a set $S$ of $n$ distinct points in the plane, a \emph{point-set embedding} of $G$ on $S$ is a drawing of $G$ where each vertex is mapped to a distinct point of $S$. Most of the literature on this topic is about planar ($\chi=0$) point-set embeddings of 
%planar 
graphs (see, e.g., \cite{biedl2012point,bose2002-outerplanar,bose1997-trees,cabello2006planar,gritzmann1991e3341,nishat2012point}). Specifically, it is NP-complete to test if a planar graph $G$ has a straight-line planar point-set embedding on a given point-set~\cite{cabello2006planar}. 
%The problem is solvable in polynomial time for outerplanar graphs and points in general position (see, e.g.,~\cite{DBLP:journals/comgeo/Bose02,bose1997-trees}). 
If $G$ is an outerplanar graph and the points are in general position such a drawing can  constructed in polynomial time (see, e.g.,~\cite{DBLP:journals/comgeo/Bose02,bose1997-trees}).

\noindent Property \ref{prop:curve-complexity} -- Together with the edge crossing minimization, the minimization of the edge bends and/or of the bends per edge are among the oldest optimization questions in graph drawing (see, e.g.,\cite{DBLP:books/ph/BattistaETT99,DBLP:conf/dagstuhl/1999dg,DBLP:reference/crc/2013gd}).
In particular, the relationship between point-set embeddings and curve complexity has been studied when the mapping between the vertices of the graph and the points is specified as part of the input as well as when it is partially specified or not specified (see, e.g., \cite{DBLP:journals/tcs/BadentGL08,DBLP:journals/jgaa/GiacomoDLMTW08,DBLP:journals/tcs/GiacomoGLN20,DBLP:journals/algorithmica/GiacomoLT10,KaufmannW02,DBLP:journals/gc/PachW01}). Notably, Kaufmann and Wiese proved that when the mapping is not specified a planar graph admits a crossing-free point-set embedding with constant curve complexity, namely at most two bends per edge~\cite{KaufmannW02}.

\noindent Property \ref{prop:crossings} -- The \emph{thrackle bound} $\vartheta(G)$ of a graph $G = (V,E)$ with $m = |E|$ edges is the number of crossings that a drawing of $G$ would have if every edge crosses every other non-adjacent edge exactly once.
It is well known that $\vartheta(G) = (m(m+1) - \sum_{v \in V}\deg^2(v))/2$; see, e.g.,~\cite{piazza1988,schaefer2013}.
As an example, for a cycle $C_n$ with $n$ vertices the above formula  gives $\vartheta(C_n) = n(n-3)/2$.
For a path $P_n$ with $n$ vertices, where all but two vertices have degree two, the formula  gives $\vartheta(P_n) = (n-2)(n-3)/2$.
A \emph{thrackle} is a drawing of a graph $G$ with exactly $\vartheta(G)$ crossings. Not all graphs admit thrackles. For example, a cycle $C_4$ with $4$ vertices does not admit a drawing with $\vartheta(C_4)= 2$ crossings. As reported in~\cite{Woodall69}, Conway conjectured in 1969 that each thrackleable graph contains at most as many edges as vertices. His conjecture,  known as  \emph{Conway's thrackle conjecture,} still remains open. 
Woodall~\cite{Woodall69} was the first to work on the thrackle conjecture and, assuming the conjecture was true, showed that a finite graph admits a thrackle if and only if it contains at most one odd cycle, no 4-cycle, and each of its connected components is either a tree or it contains exactly one cycle.  
Refer to~\cite{FulekPach2019,LovaszPS97,Xu2021} and the references therein for progress relevant to Conway's Thrackle conjecture. Graph classes that admit thrackles (and satisfy the thrackle conjecture) include cycles of more than 4 vertices~\cite{Harborth1988} and trees~\cite{piazza1988}.
Given a tree $T$ and an integer $0 \leq \chi \leq \vartheta(T)$, Piazza et al.~\cite{piazza1988} presented an algorithm that constructs a drawing of $T$ with $\chi$ crossings and curve complexity $O(n)$. A similar result was presented by Harborth~\cite{Harborth1988} for any cycle $C_n~,n>4$.  Hence, the  algorithms in~\cite{piazza1988} and~\cite{Harborth1988} compute drawings that satisfy Property~\ref{prop:crossings}.

\noindent Relationship between Properties~\ref{prop:curve-complexity} and~\ref{prop:crossings} -- A \emph{linear thrackle} is a thrackle with curve complexity zero. Not all graphs admit linear thrackles. For example, a cycle $C_6$ with $6$ edges does not admit a linear thrackle with $\vartheta(C_6)=9$ crossings (but it admits a thrackle~\cite{piazza1988}). 
It should be noted that, as proved by 
Erd\H{o}s~\cite{erdos1946} and Perles~\cite{PachS2011},  Conway's conjecture holds for linear thrackles.
%and, thus, every linear thrackle is a pseudoforest. 
More generally, there is a complex relationship between number of crossings and curve complexity. For example, a cycle $C$ with an odd number of vertices admits a drawing with curve complexity zero and $\chi$ crossings, where $\chi$ is any number up to $\vartheta(C)$ except $\vartheta(C)-1$~\cite{furry-mrcc-1977}. On the other hand, an $n$-vertex cycle $C$ with an even number of vertices admits a drawing with complexity zero and $\chi$ crossings, where $\chi$ is any number up to
$n(n-4)/2+1 = \vartheta(C)-n/2 +1$~\cite{furry-mrcc-1977}. 
Notice that there exist planar graphs admitting straight-line pointset embeddings on non-convex pointsets that have more crossings than any straight-line pointset embedding on a pointset in convex position~\cite{ChimaniFKUVW_2024}

%$\vartheta(C)-n +1$~\cite{furry-mrcc-1977}. 
% furry says "up to n(n-4)/2+1";  if you consider vartheta-x = (n(n-4)/2+1; x = n(n-3)/2 - (n(n-4)/2 - 1; x = (n^2 - 3n - n^2 + 4n) - 1;  x = n - 1

Besides the above properties we are also interested in producing RAC drawings.
A seminal user-study by Huang, Eades, and Hong~\cite{HuangEH14} shows that edge crossings that form sharp angles affect the readability of a drawing much more than those that form large angles. This has motivated the study of right-angle crossing drawings (\emph{RAC drawings}), that are graph drawings where the crossing angles are all $\frac{\pi}{2}$; see, e.g.,~\cite{Didimo2020,DidimoLM2019} for surveys and about RAC drawings. Point-set RAC embeddings were studied by Fink et al.~\cite{FinkHMSW12} who, among other results, 
show that every graph admits a point-set RAC embedding with curve complexity $3$.

\noindent{\bf Our Contribution:} We solve the main problem in this paper for any $0 \leq \chi \leq \vartheta(T)$. To this aim, we first present an $O(n)$-time algorithm which \emph{tangles} $T$ by computing a linear layout with $\vartheta(T)$ edge crossings where each edge traverses the spine (i.e. the line passing through the vertices) a constant number of times. We then show how to \emph{untangle} the linear layout until we reach the desired number $\chi$ of edge crossings. By carefully \emph{pruning} some edges of $T$, such a linear layout can be computed in $O(n^2)$ time. If $T$ is a path the time to construct the linear layout reduces to $O(n)$. Finally, the linear layout is used to construct a point-set embedding $\Gamma$ of $T$ having $\chi$ edge crossings and curve complexity at most $5$. If we require $\Gamma$ to be a RAC drawing, then $\Gamma$ has curve complexity at most $9$, which decreases to $6$ if the points are collinear. For the case of paths, a RAC point-set embedding with curve complexity $3$ can be computed in $O(n \log n)$ time.

%Due to space limitations, some technicalities/proofs appear in the appendix. 
The proofs of statements with a (clickable) ``$\star$'' appear in the appendix.

%---------------------------------------------------------------
\section{Preliminaries}\label{sec:preliminaries}
%---------------------------------------------------------------

% A \emph{graph} $G=(V,E)$ is a set $V$ of \emph{vertices} and a set $E$ of \emph{edges}, i.e., unordered pairs of vertices in $V$. 

We only consider \emph{simple} graphs, i.e.\ graphs not containing loops and multiple edges. 
%
% GDB A graph is \emph{simple} if it does not contain loops and multiple edges. In this paper we only consider simple graphs.
%
A \emph{drawing} of a graph $G(V,E)$ maps each vertex $v \in V$ to a distinct point on the plane and  each edge $(u,v) \in E$ to a Jordan arc joining the points representing $u$ and $v$. A drawing is \emph{simple} if: (i) two adjacent edges do not 
intersect except at a common end-vertex;
%have other intersections than the common end-vertex; 
(ii) no three edges intersect on the same point; (iii) two edges intersect at most once at a common interior point; and (iv) the intersection between a pair of edges must be ``transverse'', i.e., the curves of the two edges must alternate around the intersection point. We only consider simple drawings of graphs which we shall call just \emph{drawings}, for short.
A drawing is \emph{straight-line} if the edges are represented as straight-line segments. A drawing is \emph{polyline} if each edge is represented as a chain of straight-line segments such that any consecutive segments of the chain have different slopes. The \emph{curve complexity} of a polyline drawing is the maximum number of bends per edge in the drawing. A drawing is a \emph{RAC drawing} if any two edges that cross do so by forming four $\frac{\pi}{2}$ angles around the crossing point.

A \emph{topological linear embedding} $\mathcal{E}$ of graph $G(V,E)$ is defined as follows. (1) Each edge $(u,v)$ of $E$ is mapped to a \emph{subdivision path} $\pi$ composed of $k >0$ \emph{legs} $\langle (u_0,u_1), (u_1,u_2),$ \dots, $(u_{k-1},u_k) \rangle$ with $u_0=u$ and $u_k=v$, possibly consisting of a single leg $(u,v)$. The \emph{subdivision vertices} $u_1, u_2, \dots, u_{k-1}$ internal to $\pi$ are \emph{spine traversals}. (2) The vertices in $V$ and the spine traversals are ordered in $\mathcal{E}$; this order is the \emph{spine} of $\mathcal{E}$. (3) The legs are subdivided into two partition sets named \emph{top page} and \emph{bottom page} in such a way that any two consecutive legs of a subdivision path are assigned to distinct partition sets. (4) Two legs $(a,b)$ and $(c,d)$ in the same partition set \emph{cross} if $a \succ c \succ b \succ d$. (4.1) No two legs of the same edge can 
cross. (4.2) Edges $(u,v)$ and $(w,z)$ cross $\rho$ times if there are $\rho$ distinct pairs of their legs that cross. We impose $\rho \leq 1$. If $\rho=0$ %edges 
$(u,v)$ and $(w,z)$ \emph{do not cross in $\mathcal{E}$}, whereas, if $\rho=1$ they do.
%\emph{cross in $\mathcal{E}$}.

% Two edges $(u,v)$ and $(w,z)$ \emph{cross} if one leg of $(u,v)$ crosses one leg of $(w,z)$.

A drawing $\Gamma$ of a topological linear embedding $\mathcal{E}$ can be constructed as follows. The vertices and the spine traversals of $\mathcal{E}$ are represented in $\Gamma$ as points along a horizontal line, called \emph{spine} of $\Gamma$, in the same left-to-right order they have in $\mathcal{E}$. The half-plane above (resp.\ below) the spine is the \emph{top page} (resp.\ \emph{bottom page}) of $\Gamma$. A leg $(a,b)$ of $G$ assigned to the top page (resp.\ bottom page) in $\mathcal{E}$ is represented as the arc of a circumference in the top page (resp. bottom page) of $\Gamma$ having the points corresponding to $a$ and $b$ along the spine as the extreme points of the diameter. Such arc is a \emph{leg} of $\Gamma$. As a consequence, two legs cross in $\Gamma$ if and only if they cross in $\mathcal{E}$. If $a$ is a spine traversal in the topological linear embedding, the point representing $a$ along the spine of $\Gamma$ is a \emph{spine traversal} of $\Gamma$.  Drawing $\Gamma$ is a \emph{topological linear layout}
%\todo[color=green]{ASym: A topological linear layout does not necessary have semicircles as legs. It can have any curves provided the crossings are preserved.}
of $G$.  See~\cref{fig:two-page-book-embedding}. It is immediate to see that, given a topological linear embedding of a graph $G$,  a topological linear layout of $G$ can be constructed in $O(n + \sigma m)$ time, where $n$, $m$ and $\sigma$ denote the number of vertices, edges, and maximum number of spine traversals per edge of $G$, respectively.

\begin{figure}[tb]
   \centering
   \includegraphics[width=0.7\textwidth,page=7]{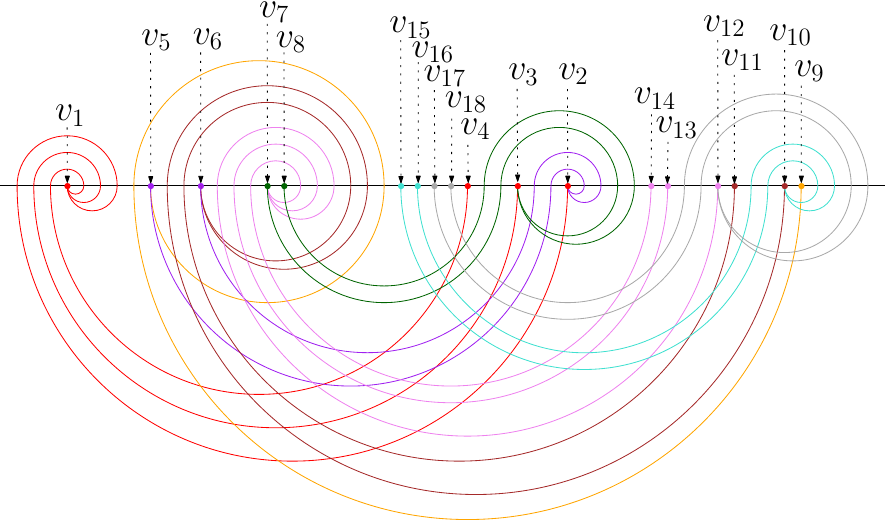}
   %\caption{Topological linear layout of the tree of \cref{fig:example-tree} with $\leq 2$ spine crossings per edge.}
   \caption{ A topological linear layout of the tree of \cref{fig:example-tree} obtained from the topological circular layout of
   \cref{fig:tree-whole}.}
   \label{fig:two-page-book-embedding}
\end{figure}

%---------------------------------------------------------------
\section{Topological Linear Embeddings of Trees}\label{sec:trees}
%---------------------------------------------------------------
In this section we describe Algorithm \textsc{Tangle}, that computes in $O(n)$ time a topological linear embedding of an $n$-vertex tree $T$ with $\vartheta(T)$ crossings and with two subdivision vertices per edge. Second, we describe Algorithm \textsc{Untangle} which can be iteratively used to reduce the number of crossings from $\vartheta(T)$ down to zero, hence guaranteeing the existence of a topological linear embedding with $\chi$ crossings, $\chi \in [1..\vartheta(T)]$ and with at most two subdivisions per edge. Since $\vartheta(T) \in O(n^2)$ and each step of \textsc{Untangle} takes $O(n)$ time, %\todo{ASymv: is this clear from the description of  the Algorithms?} 
this leads to an $O(n^3)$-time algorithm to produce a topological linear embedding with $\chi$ crossings, $\chi \in [1..\vartheta(T)]$. We will also show how to reduce this time to $O(n^2)$.

For the ease of the description, instead of producing a topological linear embedding, we produce a \emph{topological circular embedding} 
%$\mathcal{E}_\circ$ of $T$, 
whose definition is analogous to that of a topological linear embedding with the difference that the spine is a circular ordering rather than a linear ordering.
%(e.g., see \cref{fig:tree-whole}). 
A topological linear embedding is obtained from a topological circular embedding and vice versa by considering the spine ordering as linear or circular, respectively. 
Observe that crossings among the edges and spine traversals are preserved in the two embeddings.

%%%%%%
%%%%%%
%%%%%%
\subsection{Algorithm \textsc{Tangle}}

In the topological circular embedding $\mathcal{E}_\circ$ of tree $T(V,E)$ constructed by Algorithm \textsc{Tangle} each edge $e \in E$ has two subdivision vertices and, hence, it is split into three legs denoted by $\firstleg{e}, \secondleg{e}$, and $\thirdleg{e}$, where $\firstleg{e}$ is incident to the parent vertex and $\thirdleg{e}$ is incident to the child vertex. We will assign $\firstleg{e}$ and $\thirdleg{e}$ to the inner region (bottom page) and second leg $\secondleg{e}$ to the outer region (top page). The crossing among legs will happen exclusively in the inner region.

We assume that $T$ is rooted at vertex $v_1$ and planarly embedded in such a way that for any vertex $v_i$, with $i=1, \dots, n$, the non-leaf children of $v_i$ precede the leaf children of $v_i$ (see, for example, \cref{fig:example-tree}). Let $h$ be the height of $T$. For $j=0,\dots,h$, we denote by $V_j$ the set of the vertices at depth $j$ in $T$ and by $E_j$, $j=0,\dots,h-1$, the edges between the vertices in $V_j$ and those in $V_{j+1}$.
We choose a sequence $P = \langle p_1, p_2, \dots, p_{|V|+2|E|}\rangle$ of $|V|+2|E|$ distinct points in this clockwise order along a cycle $\mathcal{C}$. The points in $P$ will be used either by the vertices in $V$ or by the subdivision vertices of the edges in $E$. 
We subdivide the points in $P$ into contiguous subsequences of points along $\mathcal{C}$. More specifically, if $h$ is even, we subdivide $P$ into $P_0$, $P_2$,  $P_4$, \ldots,  $P_h$, $P_1$, $P_3$, $P_5$, \ldots, $P_{h-1}$; otherwise, we subdivide $P$ into 
$P_0$, $P_2$,  $P_4$, \ldots,  $P_{h-1}$, $P_1$, $P_3$, $P_5$, \ldots, $P_{h}$. Each subsequence $P_j$ ($j = 0, 1, \dots, h$) has size $|P_j| = |V_j| + 2|E_j|$ and is used by the vertices in $V_j$ and by the subdivision vertices of the edges in $E_j$ (refer to \cref{fig:sets-P_j}).

\begin{figure}[tb]
    \begin{subfigure}[b]{0.45\textwidth}
      \centering
      \includegraphics[page=1,width=\textwidth]{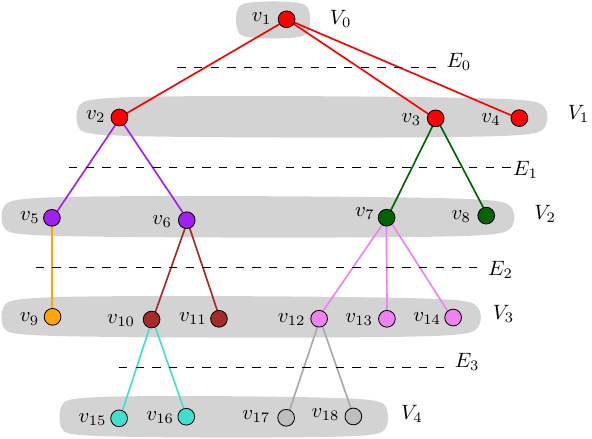}
      %\subcaption{An example tree.}
      \subcaption{}
      \label{fig:example-tree}
    \end{subfigure}\hfill
    \begin{subfigure}[b]{0.45\textwidth}
    
    \includegraphics[page=2,width=\textwidth]{example}
      %\subcaption{Subsequences $P_j$ along $\mathcal{C}$.}
      \subcaption{}
      \label{fig:sets-P_j}
    \end{subfigure}
    \caption{An example tree (a) and the subsequences of points on $\mathcal{C}$ used for its vertices/edges (b).}
    \label{fig:tree-example}
\end{figure}

\medskip
\noindent\textbf{Vertex placement.} %We are now ready to place the vertices along $\mathcal{C}$. In particular,
We map the vertices of $V_j$ to specific points of $P_j$. If $j$ is even (resp. $j$ is odd), denote by $v_{j,1}, v_{j,2}, \dots v_{j,|V_j|}$ the vertices of $V_j$ in the order (resp. in the reversed order) as they appear at depth $j$ in~$T$.  
For $k = 1, \dots, |V_j|$, we leave $\children(v_{j,k})$ points of $P_j$ unused and then we place $v_{j,k}$ (see \cref{fig:tree-vertex-placement}). Observe that after the last vertex $v_{j,|V_j|}$ of $V_j$ has been placed, there are $2|E_j|$ unused points, $|E_j|$ of which are at the end of~$P_j$.

\medskip
\noindent\textbf{Edge routing.} We describe how to embed the edges of $E_j$, for each $j = 0, \dots, h-1$. Denote by $p_{j,1}, p_{j,2}, \dots, p_{j,|V_j| + 2|E_j|}$ the points of $P_j$ in this clockwise order around $\mathcal{C}$. Recall that the vertices of $V_j$ have been already placed in $P_j$ and that each vertex $v_{j,k}$, with $k = 1, \dots, |V_j|$, is preceded by $\children(v_{j,k})$ unused points, while the whole sequence $P_j$ is closed by $|E_j|$ unused points.
For $k = 1, \dots, |V_j|$, we process vertex $v_{j,k}$ of $V_j$ and, for $\ell=1,\dots,\children(v_{j,k})$, we embed each edge $e_\ell \in |E_j|$ incident to $v_{j,k}$, where edges in $E_j$ are considered in the left to right order as they appear in the embedding of~$T$. 
Intuitively, edge $e_\ell$ will greedily use the last unused point of $P_j$ and the first unused point of $P_j$. More formally, we embed $\firstleg{e_\ell}$ in the inner region using as its endpoint the last unused point of $P_j$ (refer to \cref{fig:tree-edge-routing}); we embed $\secondleg{e_\ell}$ in the outer region using as its second endpoint the $(\children(v_{j,k})-\ell+1)$-th point of $P_j$ clockwise preceding the point used by $v_{j,k}$ (which is also the first unused point of $P_j$); we embed $\thirdleg{e_\ell}$ in the inner region hitting the point of $P_{j+1}$ where the $\ell$-th child of $v_{j,k}$ lies.

\begin{figure}[tb]
    \begin{subfigure}[b]{0.48\textwidth}
    
      \centering
      \includegraphics[page=3,width=\textwidth]{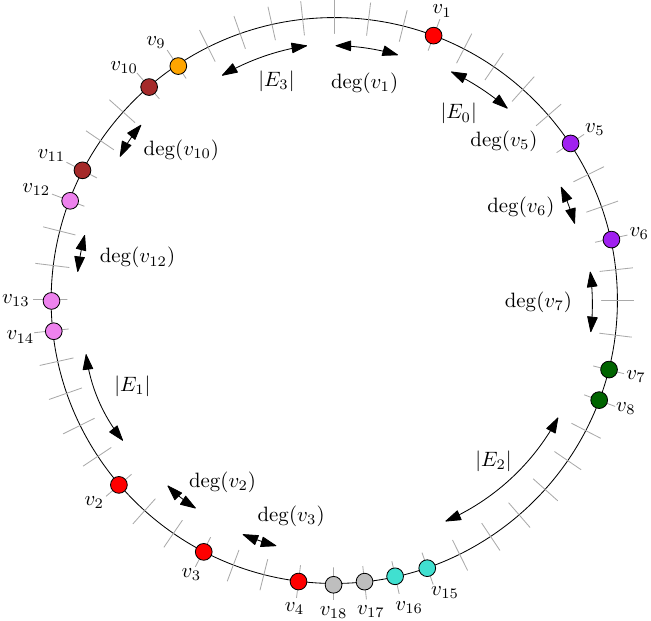}
      %\subcaption{Placement of the vertices}
      \subcaption{}
      \label{fig:tree-vertex-placement}
    \end{subfigure}\hfill
    \begin{subfigure}[b]{0.48\textwidth}

    \includegraphics[page=5,width=\textwidth]{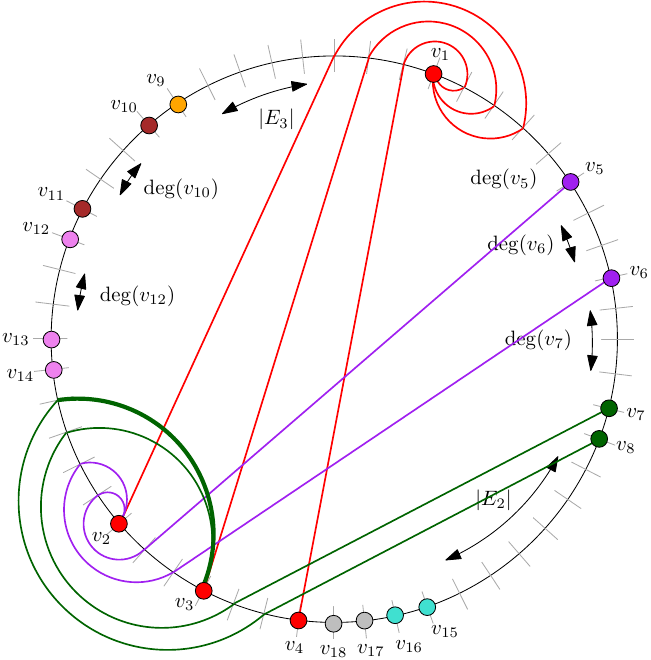}
      %\subcaption{Edge routing}
      \subcaption{}
      \label{fig:tree-edge-routing}
    \end{subfigure}
    \caption{Placement of the vertices  (a)  and  Edge routing (b) of Algorithm \textsc{Tangle}.}
    %\caption{Partial output of Algorithm \textsc{Tangle} for the tree of \cref{fig:example-tree}: (a) Placement of the vertices. (b) Edge routing of the edges in $E_0 \cup E_1$.}
    \label{fig:entangle-2}
\end{figure}

\begin{restatable}[\restateref{le:Entangle}]{lemma}{leEntangle}\label{le:Entangle}
Algorithm \textsc{Tangle} computes in $O(n)$ time a topological circular embedding of an $n$-vertex tree $T$ with $\vartheta(T)$ crossings where every edge traverses the spine exactly twice.    
\end{restatable}

\begin{sketch}
By construction, all the edges of a topological circular embedding $\mathcal{E}_\circ$ produced by Algorithm \textsc{Tangle} are composed by three legs, the second of which is drawn in the outer region. 
It can be easily shown that the second legs of the edges in $E$ do not cross in $\mathcal{E}_\circ$. Also, the first legs of the edges in $E$ do not cross among themselves in $\mathcal{E}_\circ$. 
It follows that the only possible crossings in $\mathcal{E}_\circ$ are between a first leg and a third leg of two edges in $E$ or between two third legs of two edges in $E$. 
Consider two edges $(u,v)$ and $(w,z)$ of the same $E_j$, $j \in [0..h-1]$, where $u$ and $w$ belong to $P_j$. Edges $(u,v)$ and $(w,z)$ cross whenever $u$ uses a point of $P_j$ that precedes the point of $P_j$ used by $w$ ($\firstleg{(u,v)}$ crosses $\thirdleg{(w,z)}$). Hence, all the edges of $E_j$ cross among themselves with the exception of those that are incident to the same vertex of $V_j$. 
Consider two edges $(u,v) \in E_j$ and $(w,z) \in E_{j-1}$, with $j \in [1..h-1]$, where $u$ and $z$ belong to $V_j$. Edges $(u,v)$ and $(w,z)$ cross when $u$ uses a point of $P_j$ that precedes the point of $P_j$ used by $z$ ($\firstleg{(u,v)}$ crosses $\thirdleg{(w,z)}$). 
Also, $(u,v)$ and $(w,z)$ cross when $u$ uses a point of $P_j$ that follows the point of $P_j$ used by $z$ ($\thirdleg{(u,v)}$ crosses $\thirdleg{(w,z)}$). Hence, $(u,v) \in E_j$ and $(w,z) \in E_{j-1}$ cross whenever they are not adjacent (i.e., whenever $u \neq z$).
Finally, two edges $(u,v) \in E_j$ and $(w,z) \in E_{j'}$ with $j' \notin \{j-1,j,j+1\}$ always cross once ($\thirdleg{(u,v)}$ crosses $\thirdleg{(w,z)}$). In conclusion two edges of $T$ cross if and only if they are not adjacent. Hence the number of crossings is $\vartheta(T)$.
It is easy to check that $\mathcal{E}_\circ$ can be computed in $O(n)$ time.
\end{sketch}

%ASymv
\begin{figure}[tb]
  \centering
  \includegraphics[width=0.6\textwidth,page=6]{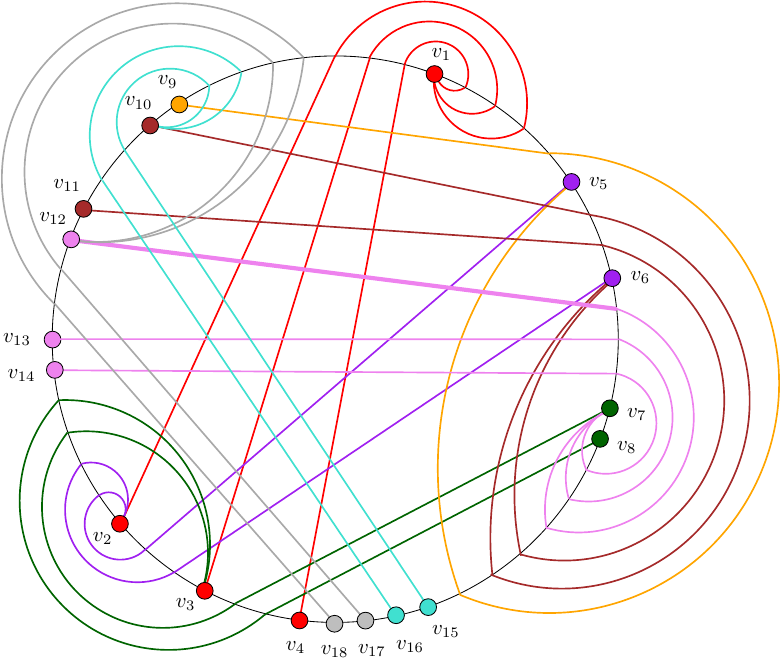}
  \caption{A drawing representing a topological circular embedding of the tree $T$ of \cref{fig:example-tree} with $\vartheta(T)$ edge crossings and two spine crossings per edge. %The topological linear layout of \cref{fig:two-page-book-embedding} is obtained from this embedding.
  }
  \label{fig:tree-whole}
\end{figure}

\begin{figure}[tb]
    \captionsetup[subfigure]{justification=centering}
    \centering
    \begin{subfigure}{0.48\textwidth}
      \centering
      \includegraphics[page=12,width=\textwidth]{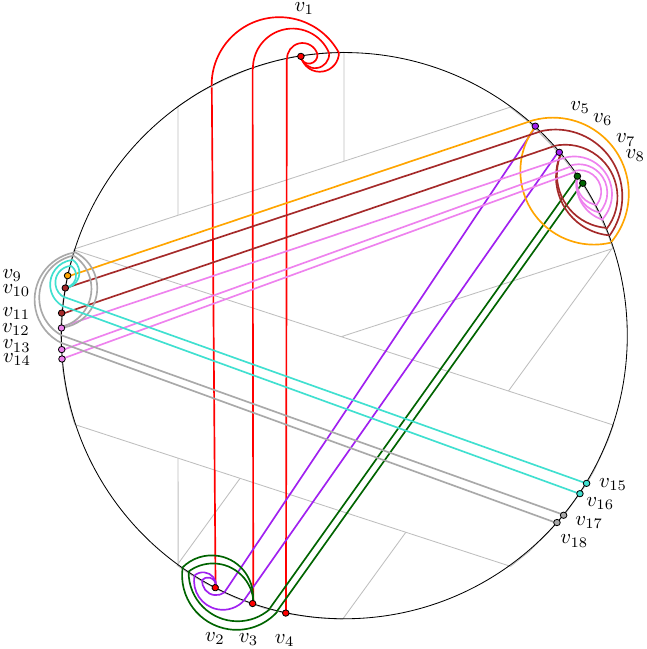}
      \subcaption{}
      \label{fig:foo3}
    \end{subfigure}\hfill
    \begin{subfigure}[b]{0.48\textwidth}
    \includegraphics[page=13,width=\textwidth]{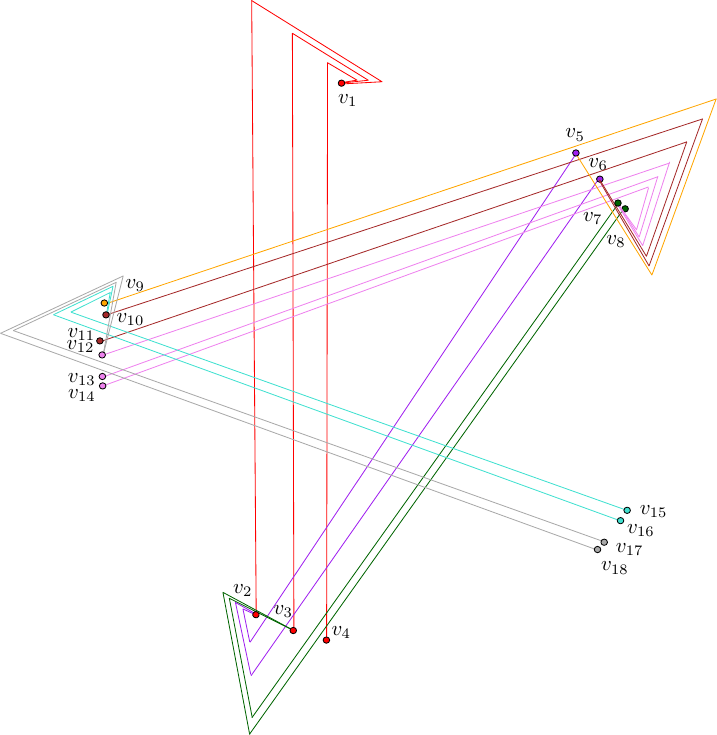}
      \subcaption{}
      \label{fig:thrackle-4}
    \end{subfigure}
    \caption{From a topological circular embedding of a tree to a thrackle with curve complexity~two.}
    \label{fig:thrackle-tree}
\end{figure}

An example drawing of the topological circular embedding computed by Algorithm \textsc{Tangle} is shown in \cref{fig:tree-whole}.  As a byproduct of the results in this section, we point out that from the circular topological embedding $\mathcal{E}_\circ$ it is not difficult to compute a polyline drawing of $T$ with $\vartheta(T)$ crossings (a thrackle of $T$) with curve complexity two as in \cref{fig:thrackle-tree}.

%%%%%%
%%%%%%
%%%%%%
\subsection{Algorithm \textsc{Untangle}}\label{se:untangle}

We describe Algorithm \textsc{Untangle} as if the target were to reduce the number of crossings down to zero. In order to obtain a topological circular embedding $\mathcal{E}^\chi_\circ$ with $\chi$ crossings it is sufficient to run $\vartheta(T)-\chi$ iterations of the algorithm. 
%{ASYMB-wrong value} Let $\vartheta'(T) = 1/2(\sum_{i=0}^{h-1}|E_i|)^2 - \sum_{i=0}^{h-2}(|E_i| \cdot |E_{i+1}|)$. 

Let $\vartheta'(T) = \frac{1}{2}\cdot|E|^2 -
\frac{1}{2}\cdot \sum_{i=0}^{h-1}|E_i|^2 - 
\sum_{i=0}^{h-2}(|E_i| \cdot |E_{i+1}|)$.
%Algorithm \textsc{Disentangle} 
The algorithm has two phases: in the first phase the 
%number of 
crossings 
%is 
are reduced from $\vartheta(T)$ to $\vartheta'(T)$, obtaining a topological circular embedding $\mathcal{E}^{\vartheta'(T)}_\circ$ where each edge is subdivided exactly once. The reason for the choice of $\vartheta'(T)$ will become evident soon. In the second phase the number of crossings is \hbox{reduced  to zero.}  
% THIS ALSO SHOWS THAT IF THE CROSSINGS ARE AT MOST $\VARTHETA'(T)$ THEN ONE SUBDIVISION PER EDGE SUFFICES

\medskip
\noindent\textbf{Phase 1: reducing the crossings to $\vartheta'(T)$.}
Let $\mathcal{E}^{\vartheta(T)}_\circ$ be a topological circular embedding of $T$ with $\vartheta(T)$ crossings. While reducing at each step the number of crossings by one, we process $T$ level by level bottom-up. In particular, for each $i = h-1, h-2, \dots, 0$, we first modify the embedding of $T$ in such a way that each vertex of $V_i$ assumes a configuration that we call ``rainbow'' and we then bring all vertices of $V_i$ and edges of $E_i$ to a configuration that we call ``full-rainbow''. When all levels of the tree are processed, the number of crossings will be exactly $\vartheta'(T)$ and each edge of $T$ will have exactly one subdivision vertex. 

A vertex $v$ is embedded as a \emph{rainbow} if: (i) the edges that lead to its children are subdivided exactly once, where their first leg %of them 
is assigned to the outer region and their second leg is assigned to the inner region and (ii) the subdivision vertex of such edges immediately precede $v$ in clockwise order.
See \cref{fig:semi-rainbow} for an example of vertices embedded as rainbows.

\begin{figure}[tb]
   \hfill
    \begin{subfigure}[b]{0.40\textwidth}
      \centering
      \includegraphics[page=1,width=\textwidth]{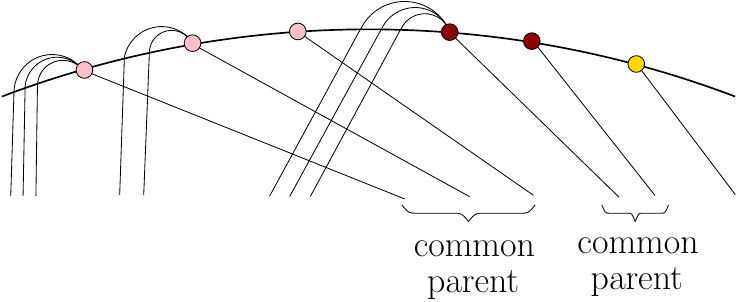}
      \subcaption{}\label{fig:semi-rainbow}
    \end{subfigure}\hfill
    \begin{subfigure}[b]{0.40\textwidth}
    
    \includegraphics[page=2,width=\textwidth]{rainbow.pdf}
      \subcaption{}\label{fig:rainbow}
    \end{subfigure}
    \hfill\hfil
    \caption{Examples of vertices embedded as rainbows (a) and a full-rainbow configuration (b).}
    \label{fig:semi-rainbow-and-rainbow}
\end{figure}
The vertices of $V_i$ are in a \emph{full-rainbow configuration} if: (i) the edges of $E_i$ are planarly drawn with only two legs, the first external and the second internal to $\mathcal{C}$ and (ii) the crossing of an edge $(u,v)$ of $E_i$ with $\mathcal{C}$ clockwise precedes all the vertices of $V_i$. See \cref{fig:rainbow} for an example of a full-rainbow configuration.
Observe that in a full-rainbow configuration of $V_i$ the first legs of edges of $E_i$ do not cross each other and no edge of $E_i$ crosses an edge of $E_{i-1}$. 
%This implies that if all $V_i$ are in full-rainbow configuration, the total number of crossings is given by the maximum number of crossings $|E|^2/2$ minus the number of possible crossings among edges of the same level $(\sum_{i=0}^{h-1}|E_i|^2)/2$ minus the number of possible crossings among edges of consecutive levels $\sum_{i=0}^{h-2}(|E_i| \cdot |E_{i+1}|)$. Notice that this number is, in fact, $\vartheta'(T)$.
%It is not hard to see that if all $V_i$ are in full-rainbow configuration the total number of crossings is $\vartheta'(T)$ (see the proof of \cref{le:tree-topological-linear-embedding-unefficient}). 

There are as many full-rainbow configurations of $V_i$ as there are %many 
permutations of the vertices in $V_i$. In particular, if we have a full-rainbow configuration and we swap the position of two consecutive vertices $v_x$ and $v_y$, we can obtain another full-rainbow configuration by swapping the position of the crossings with $\mathcal{C}$ of the edges of $E_i$ incident to $v_x$ with those of $E_i$ incident to $v_y$, and propagating the swap to the descendants of $v_x$ and $v_y$.  Hence, we obtain the following lemma:

%\begin{restatable}[\restateref{le:rainbow}]{lemma}{leRainbow}\label{le:rainbow}
\begin{restatable}{lemma}{leRainbow}\label{le:rainbow}
Let $T$ be a tree of height $h$. Let $0 \leq i \leq h$ and assume that for all $j \geq i$ the vertices of $V_j$ are in a full-rainbow configuration. There exists a topological circular embedding $\mathcal{E}_\circ$ of $T$ such that the only crossings among edges in $E_i \cup E_{i+1} \cup \dots E_{h-1}$ occur between edges in $E_j$ and $E_{k}$ with $k > j+1$.  
\end{restatable}

% \begin{property}\label{pr:rainbow}
% Let $T$ be a tree of height $h$. Let $0 \leq i \leq h$ and assume that for all $j \geq i$ the vertices of $V_j$ are in a full-rainbow configuration. There exists a topological circular embedding $\mathcal{E}_\circ$ of $T$ such that the only crossings among edges in $E_i \cup E_{i+1} \cup \dots E_{h-1}$ occur between edges in $E_j$ and $E_{k}$ with $k > j+1$.  
% \end{property}

Note that for $i = 0$ \cref{le:rainbow} yields a topological circular embedding $\mathcal{E}_\circ$ of $T$ where  the edges of the same $E_i$ do not cross and edges of $E_i$ and $E_{i+1}$ do not cross either. As indicated in~\cref{le:theta-prime}, such an embedding has  exactly $\vartheta'(T)$ crossings.

%We are now ready to describe how to decrease the number of crossings of $\mathcal{E}^{\vartheta(T)}_\circ$ down to $\vartheta'(T)$.  
Suppose to have a topological circular embedding $\mathcal{E}^{\chi}_\circ$ of $T$ such that for a given $0 \leq i \leq h-1$ it holds that: 
(i) for all $j > i$ the vertices of $V_j$ are in a full-rainbow configuration; 
(ii) for all $j \leq i$ the vertices of $V_j$ are embedded as in $\mathcal{E}^{\vartheta(T)}_\circ$; and
%ASymv:To CHECK
(iii) for all $j < i$ the edges of $E_j$ are embedded as in $\mathcal{E}^{\vartheta(T)}_\circ$.
Observe that for $i = h-1$ such a drawing is trivially $\mathcal{E}^{\vartheta(T)}_\circ$. First, we reduce the number of crossings one by one moving all vertices of $V_i$ to a rainbow configuration. Second, we reduce the crossings one by one moving $V_i$ and $E_i$ to a full-rainbow configuration. 
For the first target, recall that in $\mathcal{E}^{\vartheta(T)}_\circ$ a second leg of an edge in $E_i$ is embedded on the outer region and that the second legs of edges in $E_i$ do not cross. Consider a non-leaf vertex $u$ of $V_i$ and one of its edges $e=(u,v) \in E_i$. If $u$ has more than one child in $T$, then consider the edge $e \in E_i$ such that the first subdivision vertex $x$ of $e$ is closer to $u$ (refer to \cref{fig:disentangle}). There are three cases:
Case 1: If the last point of $P_i$ before $x$ in clockwise order is used by $u$ (see \cref{fig:disentangle-a}), then the first subdivision vertex of the edges joining $u$ with its children can be removed, and $u$ can be joined directly to their second subdivision vertices on the outer region without changing the total number of crossings (see \cref{fig:disentangle-b}). This yields an embedding of $u$ as a rainbow. 
Case 2: If the last point of $P_i$ before $x$ in clockwise order is used by a leaf $w$ (see \cref{fig:disentangle-c}), then we can swap the position of $w$ and $x$, decreasing the number of crossings by one, having removed the crossing that was between the edge incident to $w$ and $e$ (see \cref{fig:disentangle-d}). This swap operation can be repeated with the other subdivision vertices $x'$, $x''$, \dots of the edges joining $u$ with its children, reducing each time the number of crossings by one (see, for example, \cref{fig:disentangle-e}).
Case 3: If the last point of $P_i$ before $x$ in clockwise order is used by a non-leaf vertex $w$, then we can assume that $w$ is embedded as a rainbow (see \cref{fig:disentangle-f}). Indeed, if this was not the case we can process $w$ before $u$. We move $w$, together with the first subdivision vertices of the edges joining $w$ with its children, clockwise after $x$ (see \cref{fig:disentangle-g}). In doing so, we reduce the number of crossings by $\deg(w)$ instead of one unit (compare \cref{fig:disentangle-f,fig:disentangle-g} where $(u,v)$ loses four crossings). However, since the vertices of $V_{i+1}$ are in a full-rainbow configuration, they can be permuted and, in particular, $v$ can be moved counter-clockwise so that the third leg of $e$ crosses $\children(w)$ edges of $E_i$ and the total number of crossings is reduced by one unit only (see \cref{fig:disentangle-h}). Successively, we move $v$ back to its original position reducing the crossings of one unit at each swap (\cref{fig:disentangle-i} shows the first step of such a process). 
We iterate these operations until none of the above moves apply, i.e., until all the vertices of $V_i$ are in a rainbow configuration.

\begin{figure}[tb]
    \hfill
    \begin{subfigure}[b]{0.15\textwidth}

      \centering
      \includegraphics[page=1,scale=0.4]{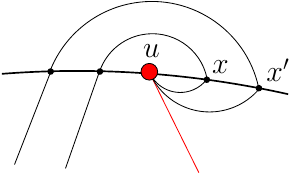}
      \subcaption{}
      \label{fig:disentangle-a}
    \end{subfigure}
    \hfill
    \begin{subfigure}[b]{0.15\textwidth}
    
    \includegraphics[page=2,scale=0.4]{disentangle.pdf}
      \subcaption{}
      \label{fig:disentangle-b}
    \end{subfigure}
    \hfill
    \begin{subfigure}[b]{0.18\textwidth}

    \includegraphics[page=3,scale=0.4]{disentangle.pdf}
      \subcaption{}
      \label{fig:disentangle-c}
    \end{subfigure}
    \hfill
    \begin{subfigure}[b]{0.18\textwidth}
    
    \includegraphics[page=4,scale=0.4]{disentangle.pdf}
      \subcaption{}
      \label{fig:disentangle-d}
    \end{subfigure}
    \hfill
    \begin{subfigure}[b]{0.18\textwidth}
    
    \includegraphics[page=5,scale=0.4]{disentangle.pdf}
      \subcaption{}
      \label{fig:disentangle-e}
    \end{subfigure}
    \hfill
    \begin{subfigure}[b]{0.24\textwidth}
    
    \includegraphics[page=6,scale=0.4]{disentangle.pdf}
      \subcaption{}
      \label{fig:disentangle-f}
    \end{subfigure}
    \hfill
    \begin{subfigure}[b]{0.24\textwidth}
    
    \includegraphics[page=7,scale=0.4]{disentangle.pdf}
      \subcaption{}
      \label{fig:disentangle-g}
    \end{subfigure}
    \hfill
    \begin{subfigure}[b]{0.24\textwidth}
    
    \includegraphics[page=8,scale=0.4]{disentangle.pdf}
      \subcaption{}
      \label{fig:disentangle-h}
    \end{subfigure}
    \hfill
    \begin{subfigure}[b]{0.24\textwidth}
    
    \includegraphics[page=9,scale=0.4]{disentangle.pdf}
      \subcaption{}
      \label{fig:disentangle-i}
    \end{subfigure}
    
    \caption{Illustrations for Algorithm \textsc{Untangle}.}
    \label{fig:disentangle}
\end{figure}

We now describe how to change the embedding of $T$ in such a way that the vertices in $V_i$ and the edges in $E_i$ are in a full-rainbow configuration. 
Recall that, since each vertex of $V_i$ has a rainbow configuration, each edge of $E_i$ has a single subdivision vertex. Consider the first subdivision vertex $x$, in clockwise order, of an edge $e \in E_i$ that is preceded by a vertex $u \in V_i$. Move $x$ before $u$ and before all subdivision vertices of the edges joining $u$ with its children. Observe that this change reduces the number of crossings by one (the crossing between $e$ and the edge of $E_{i-1}$ incident to $u$). By iterating this process we have that no subdivision vertex of $E_i$ is preceded by a vertex of $V_i$ in $P_i$, i.e., the vertices of $V_i$ are in a full-rainbow configuration.

\begin{restatable}[\restateref{le:theta-prime}]{lemma}{leThetaPrime}\label{le:theta-prime}
%Let $T$ be a tree of height $h$. 
Following the completion of Phase 1, when  all $V_i$, $0 \leq i \leq h$, are in full-rainbow configuration, the total number of crossings is $\vartheta'(T)$.  
\end{restatable}

\medskip
\noindent\textbf{Phase 2: reducing the crossings to zero.}
When all the levels of the tree are in a full-rainbow configuration, all edges are drawn with a single subdivision vertex. Consider an edge $(u,v)$ incident to a vertex $v \in V_h$ such that $v$ is the last vertex clockwise of $V_h$. Observe that $v$ is a leaf. Denote by $x$ the subdivision vertex of $(u,v)$. Vertex $v$ can be moved clockwise towards $x$ removing one crossing at each step (see, for example, \cref{fig:phase-2} in the Appendix). When $v$ reaches $x$, vertex $v$ can be transferred close to its parent $u$ and identified with it (i.e., $v$ can be planarly embedded arbitrarily close to $u$ and moved with $u$ from now on). Iteratively performing this operation on the last vertex of the last level of $T$ reduces the number of crossings down to zero, thus we have:

%Hence we have the following lemma.
\begin{restatable}[\restateref{le:tree-topological-linear-embedding-unefficient}]{lemma}{leDisentangle}\label{le:tree-topological-linear-embedding-unefficient}
Any $n$-vertex tree $T$ admits a topological linear embedding with $\chi$ edge crossings,   $\chi \in [0..\vartheta(T)]$, where every edge traverses the spine at most twice. Such an embedding can be computed in $O(n^3)$ time. 
%\end{lemma}
\end{restatable}
% \begin{proof}
% Let $T$ be a tree. We assume that the children of each vertex of $T$ are ordered in such a way that the leaves follow the other      
% \end{proof}

%%%%%%
%%%%%%
%%%%%%
\subsection{Algorithm \textsc{Prune-Tangle-Untangle}}

%As discussed, the strategy of producing a topological circular embedding $\mathcal{E}^{\vartheta(T)}_\circ$ of an $n$-vertex tree $T$ with $\vartheta(T)$ crossings and then removing the crossings one by one until an embedding $\mathcal{E}^{\chi}_\circ$ is obtained yields an $O(n^3)$-time algorithm. 
In this section we describe an algorithm, %that we call 
called
\textsc{Prune-Tangle-Untangle}, that reduces the time complexity of \cref{le:tree-topological-linear-embedding-unefficient} to $O(n^2)$.  
The algorithm uses the  \textsc{Tangle} and \textsc{Untangle} algorithms. However, before these, it launches the \textsc{Prune} procedure illustrated below. 

Let $T_n$ be the input $n$-vertex tree. The leaves of $T_n$ are recursively removed (in any order) obtaining trees $T_{n-1}$, $T_{n-2}$, \dots, while $\vartheta(T_{n-i}) > \chi$. Note that an edge $(u,v)$ of an $n$-vertex tree $T_n$, incident to a leaf $v$, may cross all other $n-2$ edges of $T_n$ except the $\deg(u) -1$ edges incident to $u$. Hence, removing $(u,v)$ produces a tree $T_{n-1}$ where the number of possible crossings decreases by a quantity in $[0..n]$. To efficiently identify the leaves and compute the current number of crossings, we equip vertices with their degrees and maintain a list of all the leaves. When a leaf is removed from the tree (and from the list) the degree of its parent is decreased, and if it becomes one, the parent is added to the list of leaves. 
We continue until we encounter a leaf that cannot be removed. Thus, it holds that $\vartheta(T_{n-i}) < \chi+n$.
%$\vartheta(T^{n-i}) - \chi \in O(n)$. 
%We continue until, for all leaves $l$ of $T^{n-i}$, we have that $\vartheta(T^{n-i}\textbackslash l) < \chi.$
%Therefore, 
Then, we launch Algorithm \textsc{Tangle} on $T_{n-i}$ to obtain $\mathcal{E}^{\vartheta(T_{n-i})}_\circ$ of $T_{n-i}$ and then perform $\vartheta(T_{n-i}) - \chi$ iterations of \textsc{Untangle} to obtain $\mathcal{E}^{\chi}_\circ$ of $T_{n-i}$. Finally, the removed edges are planarly added in reverse order as they were removed, iteratively placing them along $\mathcal{C}$ at a small enough distance from their parent, to obtain a topological circular embedding of $T$ with $\chi$ crossings. 
The above description implies the following:

\begin{theorem}\label{le:tree-topological-linear-embedding}
Any $n$-vertex tree $T$ admits a topological linear embedding with $\chi$ edge crossings, $\chi \in [0..\vartheta(T)]$, where every edge traverses the spine at most twice. Such an embedding can be computed in $O(n^2)$ time. 
\end{theorem}

%---------------------------------------------------------------
\section{Topological Linear Embeddings of Paths}\label{sec:paths} 
%---------------------------------------------------------------

In 
the first part of this section we present an algorithm, which we call \emph{$\chi$-Spiral} that, given an $n$-vertex ($n \geq 4$) path $P_n$ %\todo[color=green]{ASymv: We should describe a path based on the  number of VERTICes, not edges. }
and any integer $\chi$ such that $(n-3)(n-4)/2 < \chi \leq (n-2)(n-3)/2$ computes a spine traversal free, one-page topological linear embedding $\Gamma$ with $\chi$ crossings of $P_n$. In all topological linear embedding constructed in this section each edge is mapped to just one leg. Hence, we refer to edges and not to legs.
Notice also that $(n-3)(n-4)/2 = \vartheta(P_{n-1})$ is exactly the maximum number of crossings for a drawing of a path with $n-1$ vertices. 

Algorithm $\chi$-Spiral consists of two steps.

\noindent({\bf Step 1}) Compute a spine traversal free, one-page topological linear embedding $\Gamma$ with $\chi'=\vartheta(P_n)=(n-2)(n-3)/2$ crossings of $P_n$.
Let $v_1, \dots, v_n$ be the vertices of $P_n$ each with a subscript that corresponds to the order in which it appears by visiting $P_n$ from any of its vertices of degree one to the other. Place on the spine of $\Gamma$ first the vertices with odd subscript, in increasing order of subscript, and then the vertices with even subscript, in increasing order of subscript, too. Hence, the order of the vertices on the spine of $\Gamma$ is $v_1 \prec v_3 \prec \dots \prec v_{n-1} \prec v_2 \prec v_4 \prec \dots \prec v_n$ if $n$ is even and $v_1 \prec v_3 \prec \dots \prec v_n \prec v_2 \prec v_4 \prec \dots \prec v_{n-1}$ if $n$ is odd.
Then, assign all the edges of $P_n$ to the same page, say top page, of $\Gamma$.
The linear embedding computed by Step 1 of Algorithm $\chi$-Spiral for path $P_7$ is shown in \cref{fig:spiral}.

\begin{figure}[tb]
    \begin{subfigure}[b]{0.29\textwidth}
      \centering
      \includegraphics[page=1,width=\textwidth]{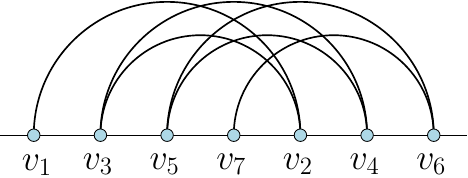}
      %\subcaption{A linear layout of $P_7$ with $10$ crossings.}
      \subcaption{}
      \label{fig:spiral}
    \end{subfigure}
    \hfill
    \begin{subfigure}[b]{0.29\textwidth}
    
    \includegraphics[page=2,width=\textwidth]{paths.pdf}
      %\subcaption{Linear layout obtained moving vertices $v_1$ and $v_7$ in the one in Fig.~\ref{fig:spiral}. $\Delta_\chi=2$ and $k=1$.}
      \subcaption{}
      \label{fig:chi-spiral}
    \end{subfigure}
    \hfill
    \begin{subfigure}[b]{0.38\textwidth}
    
    \includegraphics[page=3,width=\textwidth]{paths.pdf}
      %\subcaption{Linear layout computed for path $P_9$ and $\chi=8$.}
      \subcaption{}
      \label{fig:prescribed-spiral}
    \end{subfigure}
    \caption{Topological linear layouts of paths. (a) A linear layout of $P_7$ with $10$ crossings. (b) Linear layout obtained moving vertices $v_1$ and $v_7$ in the one in \cref{fig:spiral}. $\Delta_\chi=2$ and $k=1$. (c) Linear layout computed for path $P_9$ and $\chi=8$.}
    \label{fig:paths}
\end{figure}

\noindent({\bf Step 2}) Let $\Delta_\chi = \chi'- \chi$ 
% $\Delta_\chi = \vartheta(P_m) - \chi$ 
be the ``excess of crossings'' that $\Gamma$ has with respect to the target $\chi$. Observe that $\Delta_\chi < n-3$. Now, perform what follows. If $\Delta_\chi=2k+1$ is odd, then move vertex $v_1$ by $k+1$ positions on the spine to the right, leaving the position of all the other vertices unchanged. Else, ($\Delta_\chi=2k$ is even) move vertex $v_1$ on the spine by $k$ positions to the right, leaving the position of all the other vertices unchanged. Then, as a last action of the even case, move vertex $v_n$ to the left by one position. The assignment of all the edges to the same page remains unchanged.
The topological linear embedding computed by Algorithm $\chi$-Spiral for path $P_7$ and $\chi=8$ (i.e., $\Delta_\chi =2)$ is shown in \cref{fig:chi-spiral}.

\begin{restatable}[\restateref{le:path-chi-spiral}]{lemma}{lePathChiSpiral}\label{le:path-chi-spiral}
Let $P_n$ be an $n$-vertex ($n \geq 4$) path and let $\chi$ be an integer  such that $(n-3)(n-4)/2 < \chi \leq (n-2)(n-3)/2$. Algorithm \emph{$\chi$-Spiral} computes a spine traversal free, one-page topological linear embedding $\Gamma$ with $\chi$ crossings of $P_n$ in $O(n)$ time.   
\end{restatable}

\begin{sketch}
First, we prove that at the end of Step~1 Algorithm $\chi$-Spiral gets a topological linear embedding $\Gamma$ with $\chi'=(n-2)(n-3)/2$ crossings.
Since all the edges of $P_n$ are assigned to the same page, $\Gamma$ is one-page and spine traversal free. We prove that, in $\Gamma$, each edge crosses all the other edges but those incident to its end-vertices.
Consider any two edges $(v_i,v_{i+1})$ and $(v_j,v_{j+1})$ of $P_m$ such that $i$, $j$, $i+1$, and $j+1$ are pairwise different and assume, w.l.o.g., that $i<j$. We prove that $(v_i,v_{i+1})$ and $(v_j,v_{j+1})$ cross in $\Gamma$.
Suppose that $i$ and $j$ are both odd, the other cases being similar. We have that, since $i<j$, $v_i \prec v_j$. 
Also, both $i+1$ and $j+1$ are even, with $i+1<j+1$, hence, $v_{i+1} \prec v_{j+1}$. 
Further, since $j$ is odd and $i+1$ is even, $v_j$ is placed before $v_{i+1}$ on the spine ($v_j \prec v_{i+1}$). 
In summary, $v_i \prec v_j \prec v_{i+1} \prec v_{j+1}$. It follows that $(v_i,v_{i+1})$ and $(v_j,v_{j+1})$ cross. 
The time complexity of Step~1 of the algorithm follows from the construction, which is just a scan of the path.

Concerning Step~2 of the algorithm, we have what follows. 
(1) Because of the range of $\chi$, we have that $\Delta_\chi < n-3$. 
(2) If $\Delta_\chi$ is odd we have $0 \leq k < (n-4)/2$.
If $\Delta_\chi$ is even we have $0 \leq k < (n-3)/2$.
Because of these upper bounds on $k$, after the moves, $v_1$ is still to the left of $v_n$ in $\Gamma$. 
(3) If $v_1$ moves by one position to the right, then it is placed right after $v_3$, hence, the number of crossings in $\Gamma$ decreases by one (namely $(v_1,v_2)$ no longer crosses $(v_3,v_4)$).
(4) If $v_1$ moves by $k+1$ positions to the right then the number of crossings in $\Gamma$ decreases by $1+2k$; e.g.\ if $k=1$ then $v_1$ is positioned between $v_5$ and $v_7$, hence, $(v_1,v_2)$ no longer crosses $(v_3,v_4)$, $(v_4,v_5)$, and $(v_5,v_6)$.
Essentially, the first move of $v_1$ to the right decreases the number of crossings in $\Gamma$ by one, while each of the subsequent moves decreases the number of crossings by exactly two.
Observe also that a move of $v_n$ to the left of one position decreases the number of crossings in $\Gamma$ by one unit.
Since Step~2 involves only a linear number of moves on the spine of $\Gamma$, its time complexity is also linear.
\end{sketch}

Now, we exploit Algorithm $\chi$-spiral and \cref{le:path-chi-spiral} to efficiently produce a drawing of $P_n$ with ``any'' number of crossings, proving the following lemma.

\begin{restatable}{lemma}{lePathTopologicalLinear}\label{le:path-topological-linear}
%\begin{restatable}[\restateref{le:path-topological-linear}]{lemma}{lePathTopologicalLinear}\label{le:path-topological-linear}
%\begin{lemma}\label{le:path-topological-linear}
Let $P_n$ be an $n$-vertex path. For each $\chi \in [0..\vartheta(P_n)]$, $P_n$ admits a one-page topological linear embedding with $\chi$ crossings, that can be computed in $O(n)$ time.
%\end{lemma}
\end{restatable}

\begin{proof}
In order to compute the desired linear embedding, we first look for an integer $n'$ such that $\vartheta(P_{n'-1}) < \chi \leq \vartheta(P_{n'})$. Namely, we look for the shortest path $P_{n'}$ such that Lemma~\ref{le:path-chi-spiral} can be used to obtain a linear embedding of $P_{n'}$ with the desired number of crossings. This is given by $n'=\lceil \frac{5+\sqrt{25-4(6-2 \chi)}}{2} \rceil$.

We then execute Algorithm $\chi$-Spiral to construct a one-page linear embedding $\Gamma'$ with $\chi$ crossings of path $P_{n'}$. 
Once the one-page linear embedding $\Gamma'$ of $P_{n'}$ has been computed, we augment $\Gamma'$ by adding to it a path $P_{n''}$ with $n''=n-n'$ vertices on the same page of $P_{n'}$ just besides vertex $v_1$, positioning its vertices consecutively say before $v_1$ and joining its last vertex with $v_1$. Let $\Gamma$ be the obtained linear embedding. We have that the edges of $P_{n''}$ do not cross any edge, hence the total number of crossings of $\Gamma$ is the same of $\Gamma'$.
As an example, the linear layout computed for path $P_9$ and $\chi=8$ is shown in \cref{fig:prescribed-spiral}. In this case $n'=7$.

The $O(n)$ time bound easily descends from the description of the algorithm.
\end{proof}
%\end{comment}

%---------------------------------------------------------------
\section{Point-set Embeddings with Prescribed Edge Crossings}\label{sec:rac}
%---------------------------------------------------------------

 %In this section we describe a general approach to compute a  point-set embedding of a graph $G$ with prescribed edge crossings and $O(1)$ curve complexity.  
 In this section we address the general problem of computing a  point-set embedding of a graph $G$ with prescribed edge crossings and $O(1)$ curve complexity.  
 The following lemma extends and refines  ideas described in the literature for crossing-free point-set embeddings,~\cite{FinkHMSW12,DBLP:journals/jgaa/KaufmannW02}. 
%(see, e.g., ~\cite{FinkHMSW12,DBLP:journals/jgaa/KaufmannW02}).
 % The point-set embedding also guarantees  that the crossings form right angles. As we will see, the curve complexity can be slightly reduced if we drop the right-angle-crossing requirement. 

\begin{restatable}[\restateref{le:RAC-pse}]{lemma}{leRACpse}\label{le:RAC-pse}
  Let $G$ be a graph with $n$ vertices and $m$ edges. Let $S$ be a set of $n$ distinct points in the plane. If $G$ admits a topological linear embedding with $\chi$ edge crossings and at most $\sigma$ spine traversals per edge, then the following hold:
    \begin{itemize}
        \item $G$ admits a RAC point-set embedding on $S$ with $\chi$ edge crossings and at most $3(\sigma + 1)$ curve complexity, which  can be computed in $O(m\sigma + n\log n)$ time.
        \item $G$ admits a  point-set embedding on $S$ with $\chi$ edge crossings and at most $2\sigma + 1$ curve complexity, which can be computed in $O(m\sigma + n \log n)$ time. 
    \end{itemize}   
\end{restatable}

\begin{sketch}
    We assume, w.l.o.g., that no two points of $S$ have the same $x$-coordinate (this is in fact always the case by suitably rotating the plane).  Let $L_{top}$ be 
   %the horizontal line  through the  point of $S$ with the largest $y$-coordinate 
   a horizontal line above the point of $S$ with the largest $y$-coordinate and $L_{bottom}$ be
   %the horizontal line through the  point of $S$ with the smallest $y$-coordinate. 
   a horizontal line below the point of $S$ with the smallest $y$-coordinate.
   Let $\mathcal{S}$ be the strip of plane between $L_{top}$ and $L_{bottom}$. Project the points of $S$ onto the $x$-axis and order them by increasing $x$-coordinate. Let $\Gamma$ be a topological linear layout of $G$ whose spine is the $x$-axis and whose vertices are the projection points of $S$ onto the spine. For each spine crossing of $\Gamma$ add a dummy point to $S$ that has the same $x$-coordinate of the spine crossing and is in the interior of $\mathcal{S}$. Let $S'$ be the set of points that includes all points of $S$ and all dummy points corresponding to the spine crossings of the edges of $\Gamma$.
   In this sketch we show how to construct a RAC point-set embedding, see the appendix for the complete proof.

 Since no two points of $S'$ have the same $x$-coordinate, we can define around each point $s \in S'$ a vertical strip $\tau_s$ whose interior contains $s$ and such that no other element of $S'$ is in the interior or on the boundary of $\tau_s$; we call $\tau_s$ the \emph{safe strip of $s$}.

    Let $u$ be a vertex of $\Gamma$ and let $s\in S'$ be the point of $S'$ having the same $x$-coordinate as $u$. Let $deg_{top}(u)$ be the number of edges of $\Gamma$ incident to $u$ in the top page and let $deg_{bottom}(u)$ be the number of edges incident to $v$ in the bottom page. We place $deg_{bottom}(u)$ dummy points in the interior of $\tau_s \cap L_{bottom}$ and $deg_{top}(u)$ dummy points in the interior of $\tau_s \cap L_{top}$. We then connect $s$ to each such dummy points. Each  segment connecting $s$ to a dummy point along $\tau_s \cap L_{bottom}$ is a \emph{bottom stub} of $s$. Each  segment connecting $s$ to a dummy point along $\tau_s \cap L_{top}$ is a \emph{top stub} of $s$.   
    Let $a$ be a spine crossing along an edge of $\Gamma$ and let  $s\in S'$ be the point of $S'$ having the same $x$-coordinate as $a$. The \emph {top stub}  of $s$ is the vertical segment  from $s$ to the %whose endpoints are $s$ and the
    projection of $s$ onto  $L_{top}$. The \emph {bottom stub}  of $s$ is the vertical segment from %whose endpoints are 
    $s$ %and the 
    to the projection of $s$ onto  $L_{bottom}$. The projection point of $s$ to $L_{top}$ ($L_{bottom})$ is referred to as the \emph{ endpoint of the top (bottom) stub } of $s$.

     We are now ready to compute a RAC point-set embedding of $G$.
     Every vertex or spine crossing of $\Gamma$ is mapped to the point of $S'$ with the same $x$-coordinate. For every vertex $u \in \Gamma$ we order its incident edges in the bottom page from left to right; similarly we order from left to right its incident edges in the top page. Let $s \in S'$ be the point onto which we map $u$: We order both its bottom stubs and its top stubs from left to right. 
     Let $(u,v)$ be the $i$-th edge of $\Gamma$ incident to $u$ (either in the bottom or top page of $\Gamma$) and the $j$-th edge incident to $v$ (either in the bottom or top page of $\Gamma$). 
    
     Let $(a,b)$ be a leg of $(u,v)$ in $\Gamma$. Let $s \in S'$ be the point to which $a$ is mapped and let $s' \in S'$ be the point to which $b$ is mapped. Without loss of generality, we assume that $a \prec b$ (the case where $a \succ b$ is handled similarly). If $a$ and $b$ are both spine crossings and  $(a,b)$ is in the top (bottom) page of $\Gamma$, we connect the endpoint $p$ of the top (bottom) stub  of $s$ to the endpoint $p'$ of the top (bottom) stub   of $s'$  with a polyline consisting of two segments: the segment incident to $p$ has slope $+1$ ($-1$), the segment incident to $p'$ has slope $-1$ ($+1$). Note that the leg $(a,b)$ is mapped to a polyline with three bends: One at $p$, another at $p'$ and a third one where the  two segments with opposite slopes meet.  The case where either $a=u$ and/or $b =v$ is treated similarly; the only difference being that the polyline representing the leg $(a,b)$ in the point-set embedding must be incident to either the $i$-th stub of the point representing $u$ or to the $j$-th stub of the point representing $v$. In this case $(a,b)$ is  also represented by a polyline having $3$ bends. 
    After every leg of $\Gamma$ has been drawn in the point-set embedding by means of the above procedure, all dummy points corresponding to the spine crossings of $\Gamma$ are removed from $S'$. By construction, every edge of $\Gamma$ that crosses the spine $k$ times (and consists of $k+1$ legs) is represented in the point-set embedding as a polyline having $3(k +1)$ bends. Therefore, if $\Gamma$ has at most $\sigma$ spine crossings per edge, the curve complexity of the point-set embedding is $3(\sigma +1)$.
   
    By construction, the stubs are in the interior of the safe strips and the safe strips do not overlap with each other. This implies that the point-set embedding has no edge crossing in the interior of $\mathcal{S}$. Also, note that the safe regions follow the left to right order of the points around which they are defined: This order is consistent with the left to right order of the vertices  and spine crossings along the spine of $\Gamma$.  It follows that a crossing occurs in the point-set embedding if and only if we have four points $s_0,s_1,s_2, s_3$ in $S'$ with $x(s_0)<x(s_1)<x(s_2)<x(s_3)$ which correspond to four vertices or spine crossings  $u_0,u_1,u_2,u_3$, respectively, such that in $\Gamma$ we have an arc of circumference with diameter $u_0u_2$, another arc of circumference with diameter $u_1u_3$, such that both legs are in the same page and and $u_0 \prec u_1 \prec u_2 \prec u_3$ in the linear order. This implies that the point-set embedding of $G$ onto $S$ has the same number $\chi$ of edge crossings as the  topological linear embedding of $G$.  
   
    Since all segments along the polylines outside $\mathcal{S}$ have either slope $+1$ or slope $-1$ and since any crossing occurs outside $\mathcal{S}$, we also have that the edge crossings form $\frac{\pi}{2}$ angles and hence the point-set embedding is a RAC drawing.

 Concerning the time complexity, the  procedure in this sketch first sorts the points of $S$ by increasing $x$-coordinate, then it maps the vertices of $G$ to the points of $S$ and then it processes an edge at a time. It spends  constant time to draw any leg. It follows that all edges of the point-set embedding are processed in $O(m\sigma)$-time, which leads to a computational complexity of $O(m\sigma + n \log n)$ to compute a (RAC) point-set embedding of $G$.
 
\end{sketch}

%ASYMV------------------------------

\Cref{le:tree-topological-linear-embedding}, \Cref{le:path-topological-linear}, and~\Cref{le:RAC-pse} imply that any tree (resp.\ path) admits a RAC point-set embedding with a prescribed number of crossings and curve complexity $9$ (resp.\ $3$). It can be proved that if the points are collinear, the curve complexity of the RAC point-set embedding can be reduced to $6$ (see \cref{fig:rac} for an example and the appendix for details). If we drop the constraint that the edge crossings form right angles, we obtain smaller curve complexities by means of the same Lemmas, as it is summarized in the following theorem.

\begin{restatable}[\restateref{th:rac}]{theorem}{thRac}\label{th:rac}
    Let $T$ be a tree with $n$ vertices, let $\chi \in [0..\vartheta(T)]$, and let $S$ be a set of $n$ distinct points in the plane. The following hold:
    \begin{itemize}
        %\item There exists an $O(n^2)$-time algorithm that computes a point-set embedding of $T$ on $S$ with $\chi$ edge crossings and curve complexity $9$ if the drawing is RAC and $5$ otherwise. Also, if $S$ is a set of collinear points, the curve complexity of the RAC point-set embedding reduces to $6$.
        \item We can  compute in $O(n^2)$-time   a point-set embedding of $T$ on $S$ with $\chi$ edge crossings and curve complexity $9$ if the drawing is RAC and $5$ otherwise. Also, if $S$ is a set of collinear points, the curve complexity of the RAC point-set embedding reduces to $6$.
        %\item If $T$ is a path, there exists an $O(n)$-time algorithm that computes a point-set embedding of $T$ on $S$ with $\chi$ edge crossings and curve complexity $3$ if the drawing is RAC and $1$ otherwise.
        \item If $T$ is a path, we can  compute in $O(n \log n)$-time a point-set embedding of $T$ on $S$ with $\chi$ edge crossings and curve complexity $3$ if the drawing is RAC and $1$ otherwise.
    \end{itemize}
\end{restatable}

\begin{figure}[tb]
  \centering
  \includegraphics[width=\textwidth,page=8]{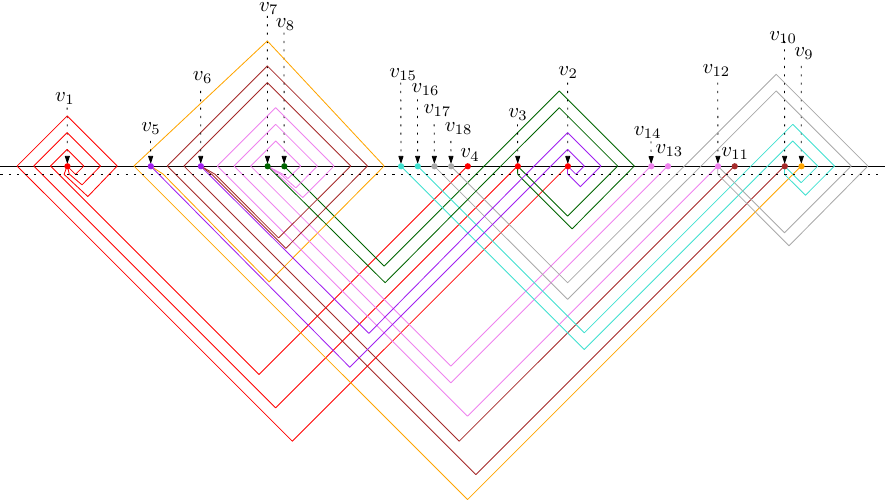}
  \caption{A RAC point-set embedding on collinear points of the tree of \cref{fig:example-tree}.}
  \label{fig:rac}
\end{figure}

%---------------------------------------------------------------
\section{Final Remarks and Open Problems}\label{sec:conclusions}
%---------------------------------------------------------------

Graph Drawing is characterized by a range of computational aesthetic metrics~\cite{DBLP:journals/access/BurchHWPWH21,DBLP:journals/vi/YoghourdjianADD18}. To investigate the impact of each metric on the readability it is useful to produce different drawings of the same graph where all metrics but one remain fixed  \cite{ddfec9eb686d402ca37a2b2a855881ff}. 
%Our paper contributes to this by producing 
In this context, we compute drawings of trees with any number of crossings (which is probably the most studied aesthetic), while guaranteeing constant curve complexity, and constraining the vertices at given locations.

We conclude with some open problems.
\begin{inparaenum}[(1)]
    \item Is there an $o(n^2)$-time algorithm to compute a point-set embedding of a tree with $\chi$ crossings and constant curve  complexity? This question is interesting even for binary trees.
    \item Is it possible to compute RAC point-set embeddings of trees with any number of edge crossings and  curve complexity smaller than $9$?
    \item A seminal paper by Pach and Wenger studies the point-set embeddability without crossings when the mapping between  vertices and  points is part of the input~\cite{DBLP:journals/gc/PachW01}. It would be interesting to study the question of the present paper also in such a setting. 
    \item Finally, it would be interesting to extend our investigation to graph classes other than trees.
\end{inparaenum}

%\newpage
\bibliography{bibliography}

\newpage
\appendix

%----------------------------------------------------------------------------------------------------------
\section{Full Proofs from \cref{sec:trees} (\nameref{sec:trees})}
%----------------------------------------------------------------------------------------------------------

First observe that a topological circular embedding can easily be used to produce a \emph{topological circular layout}, where the vertices and the spine traversals are distributed along a circle $\mathcal{C}$. We can assume that the top page corresponds to the unbounded (outer) region delimited by $\mathcal{C}$ and the bottom page corresponds to the finite (inner) region delimited by $\mathcal{C}$. Each leg $(u,v)$ can be represented as a straight segment if it is assigned to the inner region, while if it is assigned to the outer region it can be represented as an arc in the outer region on a circle $\mathcal{C'}$ that orthogonally intersect $\mathcal{C}$ at $u$ and $v$ (we may assume, up to a perturbation, that $u$ and $v$ do not correspond to a diameter of $\mathcal{C}$ and, hence, $\mathcal{C'}$ has a finite radius).

\leEntangle*
\label{le:Entangle*}
\begin{proof}
Let $\mathcal{E}_\circ$ be the topological circular embedding produced by Algorithm \textsc{Tangle}. We now determine what is the number of crossings of $\mathcal{E}_\circ$. By construction, all the edges of $\mathcal{E}_\circ$ are composed by three legs, the second of which is drawn in the outer region.

First, we show that the second legs of the edges in $E$ do not cross in $\mathcal{E}_\circ$. Observe that, since all and only the second legs of the edges in $E$ are assigned to the outer region, the second leg of an edge could only possibly cross with another second leg of some other edge.
Consider the second leg $\secondleg{e}$ of one edge $e \in E_j$, with $j \in [0..h-1]$, and the second leg $\secondleg{e'}$ of an edge $e' \in E_{j'}$, with $j' \in [0..h-1]$. If $j' \neq j$, as both endpoints of $\secondleg{e}$ belong to $P_j$ and both endpoints of $\secondleg{e'}$ belong to $P_{j'}$, the endpoints of $\secondleg{e}$ and $\secondleg{e'}$ do not interleave on the spine and, by definition, the two legs do not intersect. Otherwise, if $j' = j$, as the edges of $E_j$ greedily use the last and the first unused point of $P_j$, $\secondleg{e}$ and $\secondleg{e'}$ nest each other and do not cross in $\mathcal{E}_\circ$.  

Second, we show that the first legs of the edges in $E$ do not cross among themselves in $\mathcal{E}_\circ$. 
Consider the first leg $\firstleg{e}$ of one edge $e \in E_j$, with $j \in [0..h-1]$ and the first leg $\firstleg{e'}$ of an edge $e' \in E_{j'}$, with $j' \in [0..h-1]$. If $j' \neq j$, as both endpoints of $\firstleg{e}$ belong to $P_j$ and both endpoints of $\firstleg{e'}$ belong to $P_{j'}$, the endpoints of $\firstleg{e}$ and $\firstleg{e'}$ do not interleave along the spine and, hence, the two legs do not cross. Otherwise, if $j' = j$, observe that the starting points of $\firstleg{e}$ and $\firstleg{e'}$ are the points of $P_j$ used by the vertices in $V_j$ and, hence, their index is non-decreasing, while the ending points of $\firstleg{e}$ and $\firstleg{e'}$ are the points of $P_j$ used greedily in counter-clockwise order and, hence, their index is decreasing. Therefore, $\firstleg{e}$ and $\firstleg{e'}$ do not cross in $\mathcal{E}_\circ$.

It follows that the only possible crossings in $\mathcal{E}_\circ$ are between a first leg and a third leg of two edges in $E$ or between two third legs of two edges in $E$. 
Let's start from the first kind: let $\firstleg{e}$ be the first leg of any edge $e \in E_j$ and let $v_{j,k} \in V_j$ be the vertex $\firstleg{e}$ is incident to. Consider the third leg $\thirdleg{e'}$ of any edge $e' \in E_{j'}$. Observe that, if $j' \neq j$ and $j' \neq j-1$, then the endpoints of $\thirdleg{e'}$ are not in $P_j$ and, hence, $\firstleg{e}$ and $\thirdleg{e'}$ do not cross. 
Otherwise, if $j' = j$, $\firstleg{e}$ crosses the third leg $\thirdleg{e'}$ of any edge $e' \in E_j$ that leaves from a vertex that is placed on a point of $P_j$ that clockwise follows $v_{j,k}$. For example, in \cref{fig:tree-edge-routing} the first leg of edge $(v_3,v_7)$ (drawn thick) crosses the third leg of edges $(v_2,v_5)$ and $(v_2,v_6)$. Analogously, if $j' = j-1$, $\firstleg{e}$ crosses the third leg %$\thirdleg{e''}$ 
$\thirdleg{e'}$ 
%\todo{ASymv I do not see the reason to use $e''$. This is edge $e'$. This also applies in all occurrences of $e''$ in the rest of the proof.}
of any edge %$e'' \in E_{j-1}$ 
$e' \in E_{j-1}$ that is incident to a vertex that is placed on a point of $P_j$ that clockwise follows $v_{j,k}$. For example, in \cref{fig:tree-edge-routing} the first leg of edges $(v_3,v_7)$ and $(v_3, v_8)$ cross the third leg of edge $(v_1,v_2)$.

Finally, consider the third leg 
$\thirdleg{e}$ of one edge $e \in E_j$, with $j \in [0..h-1]$ and the third leg $\thirdleg{e'}$ of an edge $e' \in E_{j'}$, with $j' \in [0..h-1]$. If $j' = j$, $\thirdleg{e}$ and $\thirdleg{e'}$ do not cross. Indeed, assume $j$ is even (resp. $j$ is odd). Since the points in $P_j$ are numbered in the order (resp. in the reversed order) as they appear at depth $j$ in~$T$ and the points in $P_{j+1}$ are numbered in the reversed order (resp. in the order) as they appear at depth $j+1$ in $T$ we have that the two legs nest each other. If  $j' = j-1$, $\thirdleg{e}$ crosses the third leg %$\thirdleg{e''}$ 
$\thirdleg{e'}$ of any edge %$e'' \in E_{j-1}$ 
$e' \in E_{j-1}$ that is incident to a vertex that is placed on a point of $P_j$ that clockwise precedes $v_{j,k}$. For example, in \cref{fig:tree-whole} the third leg of edge $(v_7,v_{12})$ (drawn thick) crosses the third leg of edges $(v_2,v_5)$ and $(v_2,v_6)$. If  $j' = j+1$, $\thirdleg{e}$ crosses the third leg %$\thirdleg{e''}$ 
$\thirdleg{e'}$ of any edge %$e'' \in E_{j+1}$
$e' \in E_{j+1}$ that is incident to a vertex that is placed on a point of $P_{j+1}$ that clockwise follows the vertex $\thirdleg{e}$. %\todo{ASymv: I changed it to $\thirdleg{e}$ (from $\thirdleg{e''}$).} is incident to. 
%ASymv(Correct previous sentense (in {}): $\thirdleg{e''}$ is incident to.
For example, in \cref{fig:tree-whole} the third leg of edge $(v_7,v_{12})$ (drawn thick) crosses the third legs of edges $(v_{10},v_{15})$ and $(v_{10},v_{16})$. Otherwise, if $j' \notin \{j-1,j,j+1\}$, $\thirdleg{e}$ and $\thirdleg{e'}$ always cross, due to the placement of vertices in $P_i$, with $i=0,\dots, h$, along $\mathcal{C}$. 

Consider two edges $(u,v)$ and $(w,z)$ of the same $E_j$, $j \in [0..h-1]$, where $u$ and $w$ belong to $P_j$. Edges $(u,v)$ and $(w,z)$ cross whenever $u$ uses a point of $P_j$ that precedes the point of $P_j$ used by $w$ ($\firstleg{(u,v)}$ crosses $\thirdleg{(w,z)}$). Hence, all the edges of $E_j$ cross among themselves with the exception of those that are incident to the same vertex of $V_j$. 
Consider two edges $(u,v) \in E_j$ and $(w,z) \in E_{j-1}$, with $j \in [1..h-1]$, where $u$ and $z$ belong to $V_j$. Edges $(u,v)$ and $(w,z)$ cross when $u$ uses a point of $P_j$ that precedes the point of $P_j$ used by $z$ ($\firstleg{(u,v)}$ crosses $\thirdleg{(w,z)}$). 
Also, $(u,v)$ and $(w,z)$ cross when $u$ uses a point of $P_j$ that follows the point of $P_j$ used by $z$ ($\thirdleg{(u,v)}$ crosses $\thirdleg{(w,z)}$). Hence, $(u,v) \in E_j$ and $(w,z) \in E_{j-1}$ cross whenever they are not adjacent (i.e., whenever $u \neq z$).
Finally, two edges $(u,v) \in E_j$ and $(w,z) \in E_{j'}$ with $j' \notin \{j-1,j,j+1\}$ always cross once ($\thirdleg{(u,v)}$ crosses $\thirdleg{(w,z)}$). In conclusion two edges of $T$ cross if and only if they are not adjacent. This proves that the number of crossings is $\vartheta(T)$.

It is immediate to check that producing the topological circular embedding $\mathcal{E}_\circ$ takes linear time.
\end{proof}

\leThetaPrime*
\label{le:theta-prime*}
\begin{proof}
Recall that $\vartheta'(T) = |E|^2/2 -
(\sum_{i=0}^{h-1}|E_i|^2)/2 - \sum_{i=0}^{h-2}(|E_i|  |E_{i+1}|)$. In order to simplify the computations, assume that the edge sets $E_{-1}$ and $E_h$ are well defined and  empty, that is, $|E_{-1}|=|E_h|=0$.

When all $V_i,~0\leq i \leq h,$ are in full-rainbow configuration, each edge of  $E_i,~0\leq i \leq h-1$, crosses all edges of the tree but those of $E_{i-1}$, $E_i$ and $E_{i+1}$. Thus, the edges of $E_i$ are involved in a total of
$|E_i|\bigl(|E|- |E_{i-1}| - |E_i| -|E_{i+1}|\bigr)$ crossings.
Considering that each crossing is counted twice, we see that the total number of crossings is:
$$
\frac{1}{2} \sum_{i=0}^{h-1} \biggl (|E_i| \bigl(|E|- |E_{i-1}| - |E_i| -|E_{i+1}| \bigr ) \biggr) =$$ 
$$\frac{1}{2} \biggl [ |E| \sum_{i=0}^{h-1} |E_i| - 
                        \sum_{i=0}^{h-1} \bigl( |E_i||E_{i-1}|\bigr) -
                        \sum_{i=0}^{h-1} |E_i|^2 -
                        \sum_{i=0}^{h-1} \bigl( |E_i||E_{i+1}|\bigr)
\biggr] = $$
$$\frac{1}{2} \biggl [ |E|^2 - 
                        \Bigl( |E_{0}| |E_{-1}| + \sum_{i=0}^{h-2} \bigl( |E_{i+1}||E_i|\bigr) \Bigr ) -
                        \sum_{i=0}^{h-1} |E_i|^2 -
                 \Bigl( \sum_{i=0}^{h-2} \bigl( |E_i||E_{i+1}|\bigr) +|E_{h-1}|  |E_h| \Bigr)
\biggr] = $$
$$
\frac{1}{2} |E|^2 - 
\frac{1}{2} \sum_{i=0}^{h-1} |E_i|^2 - 
\sum_{i=0}^{h-2} \bigl( |E_i||E_{i+1}|\bigr)
$$
The computed quantity is precisely $\vartheta'(T)$.
\end{proof}

\leDisentangle*
\label{le:tree-topological-linear-embedding-unefficient*}
\begin{proof}
First, consider the case when $\chi=\vartheta(T)$. By  \cref{le:Entangle} Algorithm \textsc{Tangle} directly computes the required topological linear embedding in $O(n)$ time. 
Second, suppose $\vartheta'(T) \leq \chi < \vartheta(T)$, where $\vartheta'(T) = |E|^2/2 -
(\sum_{i=0}^{h-1}|E_i|^2)/2 - \sum_{i=0}^{h-2}(|E_i|  |E_{i+1}|)$. We first apply Algorithm \textsc{Tangle} and then $\vartheta(T) - \chi$ iterations of Phase 1 of Algorithm \textsc{Untangle}. 
In particular, we process $T$ bottom up removing the crossings one-by-one. As described in \cref{se:untangle}, for $i=h-1$ down to zero, the crossing removal process first embeds each vertex of $V_i$ as a rainbow and then brings all the vertices of $V_i$ in a full-rainbow configuration. 
If $\chi = \vartheta'(T)$, then all the $\vartheta(T) - \vartheta'(T)$ iteration of Phase 1 are performed, all $V_i$, $0 \leq i \leq h$, are in full-rainbow configuration and, by \cref{le:theta-prime}, the total number of crossings is $\vartheta'(T) = \chi$. Otherwise, Phase 1 of Algorithm \textsc{Untangle} is stopped when the embedding has $\chi$ crossings.  

When an iteration of Phase~1 is performed, a specific level $1 \leq i \leq h-1$ is considered and the order of the vertices of $V_i$ and of spine traversals of the edges in $E_i$ along the circle is changed by moving a single vertex or spine traversal a constant number of times, so that the number of crossings among edges in $E_i$ is decreased by one. In some specific cases (see, for example, vertex $v$ of \cref{fig:disentangle-g,fig:disentangle-h,fig:disentangle-i}) the order of the vertices of $V_{i+1}$ and spine traversals of $E_{i+1}$ is also changed, while the invariant that each level $j > i$ is in full-rainbow configuration has to be maintained. In order to do so in linear time, we associate each vertex $v$ of $V_{i+1}$ with the sub-sequence $S^v_{i+1}$ of spine traversals of the edges of $E_{i+1}$ incident to $v$ (see \cref{fig:sequences}). When a vertex $v$ is moved to a new position in the circular order of $V_{i+1}$ the sub-sequence $S^v_{i+1}$ is moved accordingly. When the position of $S^v_{i+1}$ is changed, the order of the vertices of $V_{i+2}$ has to be changed also. 
This is done by associating to $S^v_{i+1}$ the sub-sequence $V^v_{i+2}$ that are descendants of $v$ and belong to $V_{i+2}$. The change of the embedding is propagated down to the tree till the leaves are reached. More in detail, each vertex $v$ of $V_i$ is associated with sub-sequences $S^v_{i+1}$, $V^v_{i+2}$, $S^v_{i+2}$, $V^v_{i+3}$, $S^v_{i+3}$, \dots, $V^v_h$, where, for $i+1 \leq l \leq h-1$, $S^v_l$ is the sequence of spine traversals of the edges in $E_{l}$ (observe that since $l$ is in full-rainbow configuration, each edge in $E_l$ has a single spine traversal) and, for $i+2 \leq l \leq h$, $V^v_l$ is the sequence of descendants of $v$ belonging to $V_l$. 
%When $v$ is moved to the right (resp. left) of $u$ on level $i+1$, $S^v_{i+1}$ is moved to the left (resp. right) of $S^u_{i+1}$; $V^v_{i+2}$ is moved to the right (resp. left) of $V^u_{i+2}$;  $S^v_{i+1}$ is moved to the left (resp. right) of $S^u_{i+1}$; etc. (see \cref{fig:sequences} where vertex $v$ is moved to the right of $u$). 
When $v$ is moved immediately before (resp. after)  $u$ in clockwise order on level $i+1$, $S^v_{i+1}$ is moved immediately  after (resp. before)  $S^u_{i+1}$ in clockwise order; $V^v_{i+2}$ is moved immediately before (resp. after)  $V^u_{i+2}$ in clockwise order;  $S^v_{i+1}$ is moved immediately after (resp. before) of $S^u_{i+1}$ in clockwise order; etc. (see \cref{fig:sequences} where vertex $v$ is moved to the immediately before  $u$ in clockwise order).

\begin{figure}[tb]
    \begin{subfigure}[b]{0.45\textwidth}
    \includegraphics[page=10,scale=0.5]{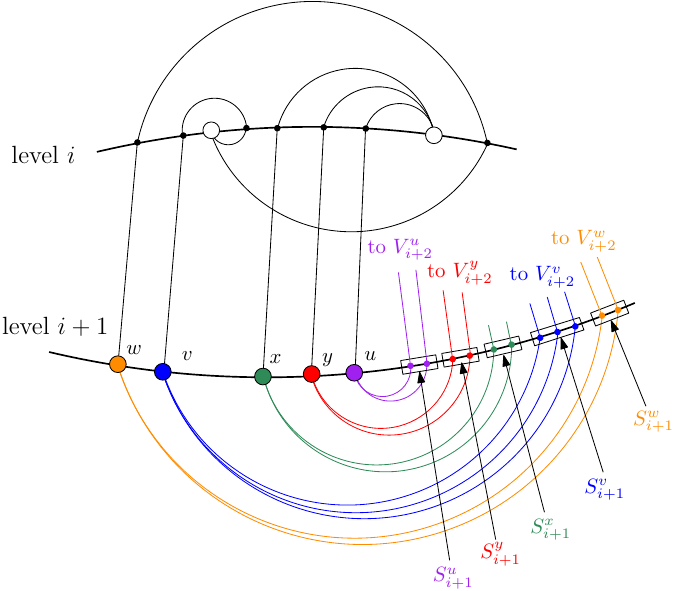}
      \subcaption{}
      \label{fig:sequences-a}
    \end{subfigure}
    \begin{subfigure}[b]{0.45\textwidth}
    \includegraphics[page=11,scale=0.5]{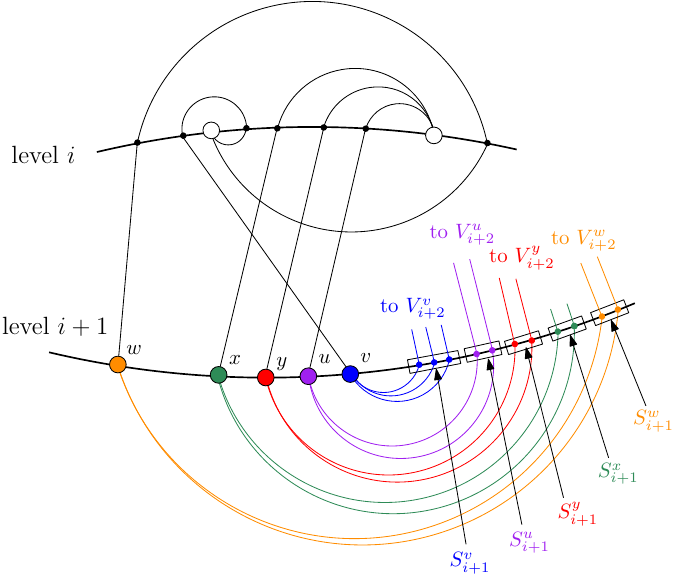}
      \subcaption{}
      \label{fig:sequences-b}
    \end{subfigure}
    \caption{Handling changes in the tree $T$ when vertex $v$ of (a) is moved to %the right of $u$ in (b).
    immediately before $u$ in clockwise order in (b).  
    }
    \label{fig:sequences}
\end{figure}

\begin{figure}[tb]
    \begin{subfigure}[b]{0.45\textwidth}
    \includegraphics[page=12,scale=0.5]{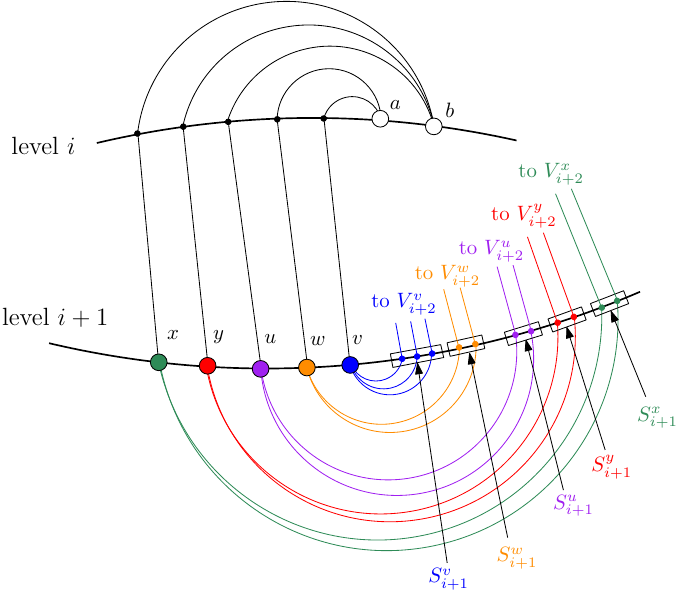}
      \subcaption{}
      \label{fig:sequences-update-a}
    \end{subfigure}
    \begin{subfigure}[b]{0.45\textwidth}
    \includegraphics[page=13,scale=0.5]{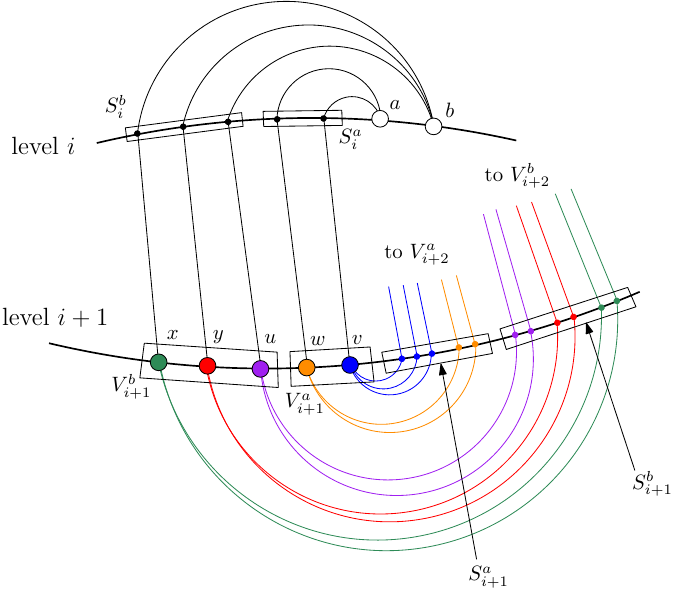}
      \subcaption{}
      \label{fig:sequences-update-b}
    \end{subfigure}
    \caption{Update of the sub-sequences when level $i$ has reached a full-rainbow configuration.}
    \label{fig:sequences-update}
\end{figure}

These changes can be performed in overall linear time and guarantee that all levels of $T$ below level $i$ are maintained in full-rainbow configuration. 
When level $i$ of $T$ is also in full-rainbow configuration (see \cref{fig:sequences-update-a}), before launching the subsequent iterations of Phase~1 of Algorithm \textsc{Untangle} on level $i-1$ of $T$, we can obtain, for each vertex $v \in V_{i}$, the above defined sub-sequences as follows.
\begin{itemize}
    \item $V^v_{i+1}$ is the sub-sequence of vertices of $V_{i+1}$ that are children of $v$ (see, for example, $V^a_{i+1}$ or $V^b_{i+1}$ of \cref{fig:sequences-update-b}).
    \item For $i+1 < l \leq h$, sub-sequence $V^v_l$ is the concatenation of sub-sequences $V^u_{l}$ for each $u$ that belongs to sub-sequence $V^v_{l-1}$ (hence, in \cref{fig:sequences-update-b}, $V^a_{i+2}$ is the concatenation of $V^v_{i+2}$ and $V^w_{i+2}$, as $v$ and $w$ are in $V^a_{i+1}$).
    \item $S^v_{i}$ is the current sub-sequence of spine traversals of the edges of $E_i$ incident to $v$ (see, for example, $S^a_{i}$ or $S^b_{i}$ of \cref{fig:sequences-update-b}).
    \item For $i+1 \leq l \leq h-1$, sub-sequence $S^v_l$ is obtained by concatenating the sequences $S^u_l$ for each $u$ that belongs to sub-sequence $V^v_l$ (see, for example, $S^a_{i+1}$ of \cref{fig:sequences-update-b} which is the concatenation of $S^v_{i+1}$ and $S^w_{i+1}$).
\end{itemize}
From the above discussion it is easy to see that the sub-sequences for all the vertices in $V_i$ can be constructed in linear time. Since the number of needed iterations of Phase~1 of Algorithm \textsc{Untangle} is $O(n^2)$, the total time is $O(n^3)$. 

Finally, suppose that $\chi < \vartheta'(T)$. We first apply Algorithm \textsc{Tangle} and then $\vartheta(T) - \vartheta'(T)$ iterations of Phase~1 of Algorithm \textsc{Untangle}. When all the levels of $T$ are in full-rainbow configuration (see, for example, \cref{fig:phase-2-a}) we perform $\chi - \vartheta'(T)$ iterations of Phase~2 of Algorithm \textsc{Untangle}, where each iteration can be performed in constant time. Again, an embedding with $\chi < \vartheta'(T)$ intersections can be obtained in $O(n^3)$ time.
%There exists a topological circular embedding $\mathcal{E}_\circ$ of $T$ such that the only crossings among edges in $E_i \cup E_{i+1} \cup \dots E_{h-1}$ occur between edges in $E_j$ and $E_{k}$ with $k > j+1$.  
\end{proof}

\begin{figure}[p]
    \hfill
    \begin{subfigure}[b]{0.30\textwidth}
    \includegraphics[page=1,scale=0.30]{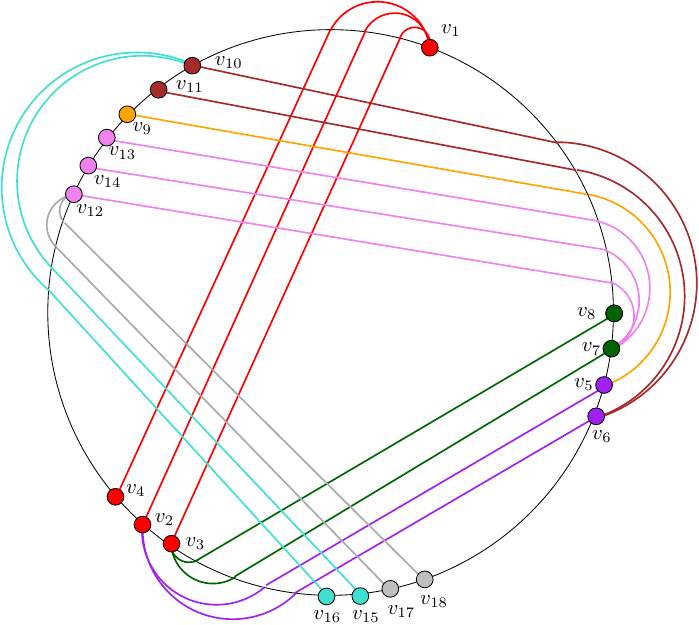}
      \subcaption{}
      \label{fig:phase-2-a}
    \end{subfigure}
    \hfill
    \begin{subfigure}[b]{0.30\textwidth}
    \includegraphics[page=2,scale=0.30]{phase-2.pdf}
      \subcaption{}
      \label{fig:phase-2-b}
    \end{subfigure}
    \hfill
    \begin{subfigure}[b]{0.30\textwidth}
    \includegraphics[page=3,scale=0.30]{phase-2.pdf}
      \subcaption{}
      \label{fig:phase-2-c}
    \end{subfigure}
     \hfill \\
     
    \begin{subfigure}[b]{0.30\textwidth}
    \includegraphics[page=4,scale=0.30]{phase-2.pdf}
      \subcaption{}
      \label{fig:phase-2-d}
    \end{subfigure} 
    \hfill
    \begin{subfigure}[b]{0.30\textwidth}
    \includegraphics[page=5,scale=0.30]{phase-2.pdf}
      \subcaption{}
      \label{fig:phase-2-e}
    \end{subfigure}
     \hfill
    \begin{subfigure}[b]{0.30\textwidth}
    \includegraphics[page=6,scale=0.30]{phase-2.pdf}
      \subcaption{}
      \label{fig:phase-2-f}
    \end{subfigure} \\
    
    \begin{subfigure}[b]{0.30\textwidth}
    \includegraphics[page=7,scale=0.30]{phase-2.pdf}
      \subcaption{}
      \label{fig:phase-2-g}
    \end{subfigure} 
    \hfill
    \begin{subfigure}[b]{0.30\textwidth}
    \includegraphics[page=8,scale=0.30]{phase-2.pdf}
      \subcaption{}
      \label{fig:phase-2-h}
    \end{subfigure}
     \hfill
    \begin{subfigure}[b]{0.30\textwidth}
    \includegraphics[page=9,scale=0.30]{phase-2.pdf}
      \subcaption{}
      \label{fig:phase-2-i}
    \end{subfigure} 
    
    \caption{Some steps of Phase~2 of Algorithm \textsc{Untangle}. (a) Every level of $T$ is in full-rainbow configuration. (b)-(h) The leaf $v_{16}$ moves towards the spine-traversal of its incident edge. (i) Leaf $v_{16}$
    is identified with its parent $v_{10}$.}
    \label{fig:phase-2}
\end{figure}

%-------------------------------------------------------

\section{Full Proofs from \cref{sec:paths} (\nameref{sec:paths})}
%\subsection{Material on Topological Linear Embeddings of Paths (in support of   \Cref{le:path-topological-linear})}\label{sec:paths} 
%---------------------------------------------------------------

\lePathChiSpiral*
\label{le:path-chi-spiral*}

\begin{proof}
First, we prove that at the end of Step~1 Algorithm $\chi$-Spiral gets a topological linear embedding $\Gamma$ with $\chi'=(n-2)(n-3)/2$ crossings.

Since all the edges of $P_n$ are assigned to the same page, $\Gamma$ is one-page and spine traversal free. We prove that, in $\Gamma$, each edge crosses all the other edges but those incident to its end-vertices.
Consider any two edges $(v_i,v_{i+1})$ and $(v_j,v_{j+1})$ of $P_m$ such that $i$, $j$, $i+1$, and $j+1$ are pairwise different and assume, w.l.o.g., that $i<j$. We prove that $(v_i,v_{i+1})$ and $(v_j,v_{j+1})$ cross in $\Gamma$.

Suppose that $i$ and $j$ are both odd. We have that, since $i<j$, $v_i \prec v_j$. 
Also, both $i+1$ and $j+1$ are even, with $i+1<j+1$, hence, $v_{i+1} \prec v_{j+1}$. 
Further, since $j$ is odd and $i+1$ is even, $v_j$ is placed before $v_{i+1}$ on the spine ($v_j \prec v_{i+1}$). 
In summary, $v_i \prec v_j \prec v_{i+1} \prec v_{j+1}$. It follows that $(v_i,v_{i+1})$ and $(v_j,v_{j+1})$ cross. 
The case $i$ and $j$ are both even is analogous.

Suppose that $i$ is odd and $j$ is even. We have that, since $i<j$ it is also $i<j+1$; since $j+1$ is odd, $v_i \prec v_{j+1}$. 
Also, both $i+1$ and $j$ are even, with $i+1<j$, hence, $v_{i+1} \prec v_j$.
Further, since $i+1$ is even and $j+1$ is odd, $v_{j+1}$ is placed before $v_{i+1}$ ($v_{j+1} \prec v_{i+1}$).
In summary, $v_i \prec v_{j+1} \prec v_{j+1} \prec v_j$. It follows that $(v_i,v_{i+1})$ and $(v_j,v_{j+1})$ cross. 
The case $i$ is even and $j$ is odd is analogous.

The time complexity of Step~1 of the algorithm follows immediately from the construction, which is just a scan of the path.

Concerning Step~2 of the algorithm, we have what follows. 
(1) Because of the range of $\chi$, we have that $\Delta_\chi < n-3$. 
(2) If $\Delta_\chi$ is odd we have $0 \leq k < (n-4)/2$.
If $\Delta_\chi$ is even we have $0 \leq k < (n-3)/2$.
Because of these upper bounds on $k$, after the moves, $v_1$ is still to the left of $v_n$ in $\Gamma$. 
(3) If $v_1$ moves by one position to the right, then it is placed right after $v_3$, hence, the number of crossings in $\Gamma$ decreases by one (namely $(v_1,v_2)$ no longer crosses $(v_3,v_4)$).
(4) If $v_1$ moves by $k+1$ positions to the right then the number of crossings in $\Gamma$ decreases by $1+2k$; e.g.\ if $k=1$ then $v_1$ is positioned between $v_5$ and $v_7$, hence, $(v_1,v_2)$ no longer crosses $(v_3,v_4)$, $(v_4,v_5)$, and $(v_5,v_6)$.
Essentially, the first move of $v_1$ to the right decreases the number of crossings in $\Gamma$ by one, while each of the subsequent moves decreases the number of crossings by exactly two.

Observe also that a move of $v_n$ to the left by one position decreases the number of crossings in $\Gamma$ by one unit.

Since Step~2 involves only a linear number of moves on the spine of $\Gamma$, its time complexity is also linear.
\end{proof}

Now, we exploit Algorithm $\chi$-Spiral and \cref{le:path-chi-spiral} to efficiently produce a drawing of $P_n$ with ``any'' number of crossings, proving the following theorem.

%----------------------------------------------------------------------------------------------------------
\section{Full Proofs from \cref{sec:rac} (\nameref{sec:rac})}
%----------------------------------------------------------------------------------------------------------
\leRACpse*
\label{le:RAC-pse*}
\begin{proof} %\todo{a figure could help}
  We assume without loss of generality that no two points of $S$ have the same $x$-coordinate (this is in fact always the case by suitably rotating the plane). Let $L_{top}$ be 
  %the horizontal line  through the  point of $S$ with the largest $y$-coordinate 
  a horizontal line  above the  point of $S$ with the largest $y$-coordinate
  and $L_{bottom}$ be
  %the horizontal line through the  point of $S$ with the smallest $y$-coordinate. 
  a horizontal line below the  point of $S$ with the smallest $y$-coordinate.
  Let $\mathcal{S}$ be the strip defined as the intersection of the bottom half-plane of $L_{top}$ and the top half-plane of $L_{bottom}$. Project the points of $S$ onto the $x$-axis and order them by increasing $x$-coordinate. Let $\Gamma$ be a topological linear layout of $G$ whose spine is the $x$-axis and whose vertices are the projection points of $S$ onto the spine. For each spine traversal of $\Gamma$ add a dummy point to $S$ that has the same $x$-coordinate of the spine traversal and is in the interior of  $\mathcal{S}$. Let $S'$ be the set of points that includes all points of $S$ and all dummy points corresponding to the spine traversals of the edges of $\Gamma$.

  We first show how to construct a RAC point-set embedding, then we drop the angle-crossing requirement, and finally discuss the time complexity.

\smallskip
\noindent{\bf RAC point-set embeddings.} Since no two points of $S'$ have the same $x$-coordinate, we can define around each point $s \in S'$ a vertical strip $\tau_s$ whose interior contains $s$ and such that no other element of $S'$ is in the interior or on the boundary of $\tau_s$; we call $\tau_s$ the \emph{safe strip of $s$}.

   Let $u$ be a vertex of $\Gamma$ and let $s\in S'$ be the point of $S'$ having the same $x$-coordinate as $u$. Let $deg_{top}(u)$ be the number of edges of $\Gamma$ incident to $u$ in the top page and let $deg_{bottom}(u)$ be the number of edges incident to $v$ in the bottom page. We place $deg_{bottom}(u)$ dummy points in the interior of $\tau_s \cap L_{bottom}$ and $deg_{top}(u)$ dummy points in the interior of $\tau_s \cap L_{top}$. We then connect $s$ to each such dummy points. Each  segment connecting $s$ to a dummy point along $\tau_s \cap L_{bottom}$ is a \emph{bottom stub} of $s$. Each  segment connecting $s$ to a dummy point along $\tau_s \cap L_{top}$ is a \emph{top stub} of $s$.   

   Let $a$ be a spine traversal along an edge of $\Gamma$ and let  $s\in S'$ be the point of $S'$ having the same $x$-coordinate as $a$. The \emph {top stub}  of $s$ is the vertical segment  from $s$ to the %whose endpoints are $s$ and the
   projection of $s$ onto  $L_{top}$. The \emph {bottom stub}  of $s$ is the vertical segment from %whose endpoints are 
   $s$ %and the 
   to the projection of $s$ onto  $L_{bottom}$. The projection point of $s$ to $L_{top}$ ($L_{bottom})$ is referred to as the \emph{ endpoint of the top (bottom) stub } of $s$.

    We are now ready to compute a RAC point-set embedding of $G$.
    Every vertex or spine traversal of $\Gamma$ is mapped to the point of $S'$ with the same $x$-coordinate. For every vertex $u \in \Gamma$ we order its incident edges in the bottom page from left to right; similarly we order from left to right its incident edges in the top page. Let $s \in S'$ be the point onto which we map $u$: We order both its bottom stubs and its top stubs from left to right. 
    Let $(u,v)$ be the $i$-th edge of $\Gamma$ incident to $u$ (either in the bottom or top page of $\Gamma$) and the $j$-th edge incident to $v$ (either in the bottom or top page of $\Gamma$). 
    
    Let $(a,b)$ be a leg of $(u,v)$ in $\Gamma$. Let $s \in S'$ be the point to which $a$ is mapped and let $s' \in S'$ be the point to which $b$ is mapped. Without loss of generality, we assume that $a \prec b$ (the case where $a \succ b$ is handled similarly). If $a$ and $b$ are both spine traversals and  $(a,b)$ is in the top (bottom) page of $\Gamma$, we connect the endpoint $p$ of the top (bottom) stub  of $s$ to the endpoint $p'$ of the top (bottom) stub   of $s'$  with a polyline consisting of two segments: the segment incident to $p$ has slope $+1$ ($-1$), the segment incident to $p'$ has slope $-1$ ($+1$). Note that the leg $(a,b)$ is mapped to a polyline with three bends: One at $p$, another at $p'$ and a third one where the  two segments with opposite slopes meet.  The case where either $a=u$ and/or $b =v$ is treated similarly; the only difference being that the polyline representing the leg $(a,b)$ in the point-set embedding must be incident to either the $i$-th stub of the point representing $u$ or to the $j$-th stub of the point representing $v$. In this case $(a,b)$ is  also represented by a polyline having $3$ bends. 
   After every leg of $\Gamma$ has been drawn in the point-set embedding by means of the above procedure, all dummy points corresponding to the spine traversals of $\Gamma$ are removed from $S'$. By construction, every edge of $\Gamma$ that traverses the spine $k$ times (and consists of $k+1$ legs) is represented in the point-set embedding as a polyline having $3(k +1)$ bends. Therefore, if $\Gamma$ has at most $\sigma$ spine traversals per edge, the curve complexity of the point-set embedding is $3(\sigma +1)$.
   
   By construction, the stubs are in the interior of the safe strips and the safe strips do not overlap with each other. This implies that the point-set embedding has no edge crossing in the interior of $\mathcal{S}$. Also, note that the safe regions follow the left to right order of the points around which they are defined: This order is consistent with the left to right order of the vertices  and spine traversals along the spine of $\Gamma$.  It follows that a crossing occurs in the point-set embedding if and only if we have four points $s_0,s_1,s_2, s_3$ in $S'$ with $x(s_0)<x(s_1)<x(s_2)<x(s_3)$ which correspond to four vertices or spine traversals  $u_0,u_1,u_2,u_3$, respectively, such that in $\Gamma$ we have an arc of circumference with diameter $u_0u_2$, another arc of circumference with diameter $u_1u_3$, such that both legs are in the same page and and $u_0 \prec u_1 \prec u_2 \prec u_3$ in the linear order. This implies that the point-set embedding of $G$ onto $S$ has the same number $\chi$ of edge crossings as the  the topological linear embedding of $G$.  
   
   Since all segments along the polylines outside $\mathcal{S}$ have either slope $+1$ or slope $-1$ and since any crossing occurs outside $\mathcal{S}$, we also have that the edge crossings form $\frac{\pi}{2}$ angles and hence the point-set embedding is a RAC drawing.

\smallskip
\noindent{\bf Non-RAC point-set embeddings.}  If we do not require that the edge crossings of the point-set embedding form right angles, we can construct a point-set embedding by refining ideas of \cite{DBLP:journals/jgaa/KaufmannW02} and of ~\cite{DBLP:journals/tcs/BadentGL08,DBLP:journals/tcs/GiacomoGLN20} which describe crossing free point-set embeddings with constant curve complexities.
More precisely, we again define $S'$ as we did in the case of RAC point-set embeddings and we  draw a path $\pi$  by connecting the points of $S'$ in the increasing order of their x-coordinates. Let $\rho$ be the maximum slope, in absolute value, of any segment along $\pi$ and let $\tau$ be an arbitrary slope such that $\tau > \rho$. As in the previous case, $\Gamma$ denotes a topological linear layout  representing the topological linear embedding of $G$.

Let $(u,v)$ be an edge of $\Gamma$ and let $(a,b)$ be a leg of $(u,v)$. Let $s \in S'$ be the point to which $a$ is mapped and let $s' \in S'$ be the point to which $b$ is mapped.  We map $(a,b)$ to a polyline that connects $s$ and $s'$ and consists of two segments, one with slope $\tau$ and the other with slope $-\tau$. If $(a,b)$ is in the top page of $\Gamma$, the polyline is drawn above $\pi$; below $\pi$ otherwise.
 Once all legs of $G$ have been represented this way, we have that by construction two edges in the point-set embedding cross if and only if they cross in $\Gamma$. However, some edges that are incident on a same vertex may partly overlap. Let $e_1, e_2, \cdots, e_k$ be a set of such overlapping edges and let $u$ be their common endvertex. We perturb the slope of each edge by a quantity that is less than $\tau - \rho$ in a such a way that: (i) no overlap is left and (ii) the circular order of the edges incident to $u$ is preserved. Finally, we remove the edges of path $\pi$ and the points of $S'$ that in $\Gamma$ correspond to spine traversals (these points  become bends).

Since: (i) each leg is represented as a polyline with one bend, (ii) each edge $(u,v)$ has at most $\sigma$ legs, and (iii) the point-set embedding can have a bend at each point of $S'$ that correspond to a spine traversals of $\Gamma$, we have that the curve complexity of the point-set embedding is at most $2\sigma +1$.  

\smallskip
\noindent{\bf Time complexity.} Concerning the time complexity, the two procedures in this proof  sort the points of $S$,  map the vertices of $G$ to the points of $S$, and  process an edge at a time. They spend constant time to draw any   portion of an edge between two consecutive spine traversals (if any). It follows that all edges of the point-set embedding are processed in $O(m\sigma)$-time, which leads to an overall computational complexity of
$O(m\sigma + n \log n)$ 
%\todo{This must be wrong! We state $O(m\sigma + n \log n)$ in the Lemma.}
to compute a (RAC) point-set embedding of $G$.
\end{proof}

\thRac*
\label{th:rac*}
\begin{proof}
    The statement about paths and the statement about trees with the points not being collinear are an immediate consequence of~\cref{le:tree-topological-linear-embedding}, \cref{le:path-topological-linear}, and~\cref{le:RAC-pse}. It remains to prove the statement in the case that $T$ is a tree with more than two leaves, $S$ is a set of collinear points, and we aim at computing a RAC point-set embedding of $T$ on $S$. We assume $T$ to be rooted which implies that for any edge $(u,v)$ of $T$ $u$ is the parent of $v$ or vice-versa. Also, w.l.o.g. we assume that the $x$-axis passes through the points of $S$.

  We sort the points of $S$ by increasing $x$-coordinate.  We compute a topological linear embedding $\mathcal{E}$ of $T$ with $\chi$ crossings and at most two spine traversals per edge in $O(n^2)$ time by ~\cref{le:tree-topological-linear-embedding}. From $\mathcal{E}$ we  construct in $O(n)$ time a topological linear layout $\Gamma$ of $T$ such that the vertices are mapped to the points of $S$ (see also \cref{sec:preliminaries}). For a vertex $u$ of $\Gamma$, let $deg_{top}(u)$ ($deg_{bottom}(u)$) be the number of its  incident edges in the top (bottom) page. % and let $deg_{bottom}(v)$ be its incident edges in the bottom page. 

   To construct a RAC point-set embedding of $T$, we replace each leg $(a,b)$ of $\Gamma$ by a polyline consisting of the two sides of an isosceles right triangle whose hypotenuse is the segment of the spine with $a$ and $b$ as its endpoints. Let $\Gamma'$ be the resulting  drawing. Although, strictly speaking, $\Gamma'$ is not a topological linear layout of $G$ we shall call spine the line through the vertices of $\Gamma'$; similarly, we call spine traversal of $\Gamma'$  the point of intersection between its spine and a segment of an edge not incident to a vertex. The top (bottom) page of $\Gamma'$ is the half-plane above (below) its spine. A leg of $\Gamma'$ is the polyline of $\Gamma'$ representing a leg of $\Gamma$.

If $T$ is not a path, %\todo[color=green]{ASymv: I rewrote the next 2 paragraphs. Minor changes in the notation used.}
$\Gamma'$ will have some overlapping edges. The edge overlaps are removed   as follows:
    Let $\delta$ be the minimum horizontal distance between any two points representing vertices or spine traversals along the spine of $\Gamma'$. For any vertex $u$ along the spine of $\Gamma'$, we define a vertical strip $\tau_u$ having width $\frac{\delta}{4}$ and such that $u$ is in the center of $\tau_u$. Let $L_{top}$ and $L_{bottom}$ be the two horizontal lines having $y$-coordinates $\frac{\delta}{16}$ and $-\frac{\delta}{16}$, respectively. Note that the $y$-coordinates  of lines $L_{top}$ and $L_{bottom}$ are selected so that, for each vertex $u$, the lines $L_{top}$ and $L_{bottom}$ intersect the edges incident to $u$ in the interior of the strip $\tau_u$. 
       For every vertex $u$, let $l^b_u = \tau_u \cap L_{bottom}$ and $l^t_u = \tau_u \cap L_{top}$   be two line segments associated with $u$, one in the bottom page and the other in the top page, respectively. We  subdivide  $l^b_u$ with $deg_{bottom}(u)$  dummy points and $l^t_u$ with  $deg_{top}(u)$ dummy points. We shall connect $u$ to each such dummy point: A  segment connecting $u$ to a dummy point along $l^b_u$ ($l^t_u$)  is a \emph{bottom stub} (\emph{top stub}) of $u$. We order the edges incident to $u$ in the bottom (top) page from left to right and associate the $j$-th edge with the $j$-th bottom (top) stub, $1 \leq j \leq deg_{bottom}(u)$ ($1 \leq j \leq deg_{top}(u)$).  

    We now provide details on how to place, for each vertex $v$ of $T$,  the subdivision points on $l^b_v$ and $l^t_v$, where the vertices are processed in a top-down traversal of $T$.  
    Let $(u,v)$ be any edge of $\Gamma'$ such that $u$ is the parent of $v$ in $T$. In the description that follows, we assume that: (i) $(u,v)$ is the $i$-th edge incident to $u$ in the bottom page; (ii) $(u,v)$ is the $k$-th edge incident to $v$ in the bottom page,  and (iii) $u$ is to the left of $v$ (the case where  either $(u,v)$ is incident to $u$ and/or to $v$ in the top page or when $u$ is to the right of $v$ is treated in a similar way, possibly  changing sign to the slopes of the segments).  The $k$-th subdivision point along $l^b_v$, say $p$,  is chosen so that  its incident bottom stub  has the same slope as the edge $(u,v)$ in $\Gamma'$ (i.e. either $+1$ or $-1$).  Then, along $l^b_v$,    $k-1$ subdivision points are placed to the left of $p$  and $deg_{bottom}(v)-k$ subdivision points are placed to the right of $p$.  The above placement of subdivision points guarantees that the unique edge entering $v$ from its parent uses a bottom stub of slope either $+1$ or $-1$.
    Consider now any leg $(a,b)$ of  edge $(u,v)$ in $\Gamma'$. If  both $a$ and $b$ are spine traversals, or if $a$ is a spine traversal and $b=v$, we do not modify the polyline representing $(a,b)$ in $\Gamma'$. If $a=u$ and either $b$ is a spine traversal or $b=v$, let $q$ be the $i$-th subdivision point of $l^b_u$. We draw $(a,b)$ as a polyline consisting of three segments: Two segments, one having slope $-1$ and the other having slope $+1$, connect $q$ to $b$; the third segment is the $i$-th stub of $u$, namely the segment whose endpoints are $u$ and $q$. By repeating this procedure for every edge of $\Gamma'$, the overlaps between the edges incident to a common vertex are removed.    
    Concerning the curve complexity of edge $(u,v)$, note that each leg has one bend except for the leg incident to $u$, which has two bends; also, there can be a bend at each spine traversal. Since $(u,v)$ has at most two spine traversals by \cref{le:tree-topological-linear-embedding}, it follows that  the point-set embedding has curve complexity $6$. 
    By construction, the edge crossings only occur between segments having opposite slopes. Therefore we have constructed a RAC point-set embedding of $T$ on $S$. \cref{fig:rac} is an example of a RAC point-set embedding on a set of collinear points computed with the procedure above. Since every leg can be processed in $O(1)$ time and the number of legs is $O(n)$, the overall time complexity is dominated by the construction of the linear embedding of $T$, that is $O(n^2)$ by \cref{le:tree-topological-linear-embedding}.  
\end{proof}

\end{document}